\documentclass[lettersize,journal]{IEEEtran}
\usepackage{amsmath,amsfonts}
\usepackage{algorithmic}
\usepackage{algorithm}
\usepackage{array}
\usepackage[caption=false,font=normalsize,labelfont=sf,textfont=sf]{subfig}
\usepackage{textcomp}
\usepackage{stfloats}
\usepackage{url}
\usepackage{verbatim}
\usepackage{graphicx}
\usepackage{cite}
\usepackage{amsthm}
\newtheorem{theorem}{Theorem}
\usepackage{soul}
\usepackage{color}

\begin{document}

\title{Achieving Interference-Free Degrees of Freedom in Cellular Networks via RIS }

\author{
	\IEEEauthorblockN{Junzhi Wang, Jun Sun$^{*}$, Zheng Xiao, Limin Liao, Yingzhuang Liu}\\
 \thanks{The authors are with The School of Electronic Information and Communications, Huazhong University of Science and Technology, Wuhan, China. E-mails: \{wangjunzhi, juns, xiaoz, liaolimin001, liuyz\}@hust.edu.cn}
} 

\maketitle

\begin{abstract}
It's widely perceived that Reconfigurable Intelligent Surfaces (RIS) cannot increase Degrees of Freedom (DoF) due to their relay nature. A notable exception is Jiang \& Yu's work. They demonstrate via simulation that in an ideal $K$-user interference channel, passive RIS can achieve the interference-free DoF.
In this paper, we investigate the DoF gain of RIS in more realistic systems, namely cellular networks, and more challenging scenarios with direct links. We prove that RIS can boost the DoF  per cell to that of the interference-free scenario even \textit{ with direct-links}. Furthermore, we \textit{theoretically} quantify the number of RIS elements required to achieve that goal, i.e. $max\left\{ {2L, (\sqrt L  + c)\eta+L } \right\}$ (where $L=GM(GM-1)$, $c$ is a constant and $\eta$ denotes the ratio of channel strength) for the $G$-cells with more single-antenna users $K$ than base station antennas $M$ per cell. 
The main challenge lies in addressing the feasibility of a system of algebraic equations, which is difficult by itself in algebraic geometry. We tackle this problem in a probabilistic way, by exploiting the randomness of the involved coefficients and addressing the problem from the perspective of extreme value statistics and convex geometry.  
Moreover, numerical results confirm the tightness of our theoretical results.

\end{abstract}

\begin{IEEEkeywords}
Reconfigurable intelligent surface, degree of freedom, cellular networks, interference nulling.
\end{IEEEkeywords}

\section{Introduction}
\IEEEPARstart{R}{econfigurable} intelligent surface (RIS) has been envisioned as one of the most promising technologies for the next generation of wireless communication systems due to its low cost and ability to adjust the propagation of signals through passive reflection \cite{wu2021intelligent,ZhouGui,long2021active,10274111}. 
Since its instantaneous reflection property, RIS can intelligently modify the channel conditions by adjusting the phase of each reflecting element. For example, it can transform a fully connected wireless system into a partially connected one by eliminating cross-interference links, thereby enhancing the system's DoF \cite{bafghi2022degrees}. Unlike conventional interference alignment, which requires symbol extension and involves high time complexity, RIS offers a more practical solution for managing interference without the need for any symbol extension. However, too few reflecting elements is insufficient for achieving the desired channel modifications, while too many will increase costs as low-cost RIS is not no-cost. Therefore, determining the optimal number of reflecting elements is crucial for implementing RIS-aided interference management in practical systems.

In this paper, we investigate a practical RIS-aided interference management in cellular networks by asking the following question:
\begin{itemize}	
	\item{Can RIS boost the DoF per cell and how large is it at most?}
	\item{How many reflective elements in RIS are required to achieve the maximum DoF?}
\end{itemize}

Addressing these questions is equivalent to evaluating the solvability of the interference-nulling equation 
\begin{equation}
\left\{ {\begin{aligned}  
  &{{\mathbf{Ax}} + {\mathbf{b}} = 0} \\ 
  &{\left| {{x_i}} \right| = 1,{\text{for each element of }}{\mathbf{x}}} 
\end{aligned}} \right.
\label{I1},
\end{equation}
where ${\mathbf{A}} \in {\mathbb{C}^{L \times N}}$ and ${\mathbf{b}} \in {\mathbb{C}^{L \times 1}}$ are complex random coefficients. Here, ${\mathbf{A}}$ and ${\mathbf{b}}$ can represent the cascaded reflective channel and the direct channel respectively, and ${\mathbf{x}} \in {\mathbb{C}^{N \times 1}}$ can represent the phase shift of RIS. 
We need to find a combination of column vectors from ${\mathbf{A}}$ that matches $-{\mathbf{b}}$ exactly in both phase and amplitude. However, the modulus-1 constraints of ${\mathbf{x}}$ make this challenging, as it restricts the columns of ${\mathbf{A}}$ to rotation only, prohibiting any stretching. 
When ${\mathbf{b}} = 0$, the authors of \cite{jiang2022interference} demonstrated that equation \eqref{I1} can be solved if $N$ slightly exceeds $2L$.
This is because the equation consists of $L$ complex equations and $N$ real variables (i.e., $N$ phase shift), requiring at least $2L$ variables for solvability. They also showed by simulation that when the magnitude of ${\mathbf{b}}$ is non-zero and exceeds a certain threshold, more variables are required. But this threshold is still unknown. Additionally, as the magnitude of ${\mathbf{b}}$ increases, this $2L$ criterion gradually diverges from the actual number of required variables.
This is because a sufficiently large ${\mathbf{b}}$ requires ${\mathbf{A}}$ to provide more columns to compensate for the amplitude disparity. Which further increases the difficulty of determining the solvability conditions of this equation. 

Given the randomness of ${\mathbf{A}}$ and ${\mathbf{b}}$, we solve it from a probability perspective. 
Specifically, if the equation \eqref{I1} hold, there must be $\left| {{\mathbf{Ax}}} \right| = \left| {\mathbf{b}} \right|$. Where $\left| {{\mathbf{Ax}}} \right|$ represents the projection of an $N$-dimensional vector ${\mathbf{x}}$ into a lower dimensional space, its magnitude can be characterized by Gordon's Theorem \cite{bandeira2015ten}. This is because, in high-dimensional spaces, the length and direction of a vector can vary significantly. However, when it is randomly projected into a lower-dimensional space, these diversities will be averaged out. As a result, its length tends to be more stable and concentrated around a certain expected value \cite{vershynin2018high}. Gordon's Theorem precisely describes this concentration effect in random projections. Therefore, We first determine the range of $\left| {{\mathbf{Ax}}} \right|$ by Gordon's Theorem and compare it with $\left| {\mathbf{b}} \right|$ to establish the solvability condition. Additionally, the extreme value statistics of the magnitude of the linear combination of ${\mathbf{A}}$'s elements and ${\mathbf{b}}$ are also explored. Furthermore, since the constrained ${\mathbf{x}}$ lies symmetrically on the surface of a ball, the essence of \eqref{I1} can be approximated as identifying the intersection between the sphere and the null-solution space of equations. Consequently, the geometric relationship between these two sets is also employed to establish the solvability conditions.

\subsection{Related Works}
The RIS is a metasurface comprising numerous reconfigurable passive reflecting elements, each capable of imposing an independent phase shift on incoming wireless signals \cite{shao2024target,wang2021joint}. Consequently, signals from different paths can be constructively combined at the desired users to boost signal power, and destructively combined at unintended users to eliminate interference. It has been demonstrated that RIS can significantly enhance both spectral and energy efficiency in various systems. For instance, \cite{pan2020multicell} showed that RIS can improve the cell-edge performance and weighted-sum rate in multi-cell multiple-input multiple-output (MIMO) systems. 
Additionally, \cite{wu2019intelligent} showed that the transmit power of the access point can be decreased in the order of $N^2$ with a sufficiently large number of RIS elements $N$.
This squared power gain is also achievable for the RIS with discrete phase shifts \cite{wu2019beamforming}.

With the surge in terminal devices and diverse service demands, high capacity and massive access have become key features of the next-generation wireless communication systems 
\cite{chen2022standardization,tang2020wireless,matthaiou2021road,ye2022reconfigurable}. 
The capacity limits of a RIS-aided multiple-input single-output (MISO) broadcast channel were initially explored in \cite{chen2023fundamental}. It revealed that with a sufficiently large number of RIS elements, the sum rate achieved by RIS-aided zero-forcing is comparable to that of the more complex dirty paper coding scheme. This implies that the complex transmission scheme can be replaced by a more easily implemented
scheme assisted by RIS. However, the minimum required number of the RIS elements is unclear. 
When numerous users access the network simultaneously, time and spectrum sharing become unavoidable. Consequently, the main challenge in enhancing system performance is effectively managing interference.

Due to the inherent physical nature of the RIS, it can reflect both the impinging signal and the undesired interference from the surrounding environment. While this interference is uncontrollable and the associated channel cannot be estimated \cite{de2021electromagnetic}. 
The authors of \cite{khaleel2023electromagnetic} proposed an electromagnetic interference (EMI) cancellation scheme to mitigate this interference by exploiting the time-domain behavior of the EMI and optimizing the RIS phase shifts across different time slots.

Compared to the interference from the environment, inter-user interference has a greater impact on the system performance, particularly in multi-user systems.
To eliminate the inter-user interference more efficiently, the authors of \cite{fu2021reconfigurable} proposed a RIS-assisted interference alignment strategy. Specifically, the RIS phase and precoding matrix were jointly optimized to improve the feasibility of interference alignment conditions.
Subsequently,  an enhanced interference alignment strategy with a minimum interference leakage (MIL) criterion for the RIS-assisted multiuser MIMO system was proposed in \cite{xu2023enhanced}. The MIL criterion was used to optimize the precoding matrix iteratively. This strategy enabled a trade-off between the capacity and the computational complexity.

Although interference alignment has been proven to be a useful interference management, it is difficult to implement in actual systems. This is because the primary operation of interference alignment is designing a precoding matrix to minimize the interference subspace, thereby maximizing the available space for useful signals  \cite{cadambe2008interference,cadambe2009interference,5502363}. However, this usually requires symbol extension, resulting in high time complexity.

Recent studies have demonstrated that, in the $k$-user interference channel, interference can be completely eliminated by RIS \cite{bafghi2022degrees,chae2022cooperative,jiang2022interference}. Consequently, a full DoF of $k$ can be achieved. Compared with the conventional interference alignment, which can only recover half of the total $k$ DoF \cite{cadambe2008interference}, the RIS can not only significantly increase the system's DoF but also does not require any symbol expansion. Therefore, the RIS-assisted interference nulling strategy is more practical.

The time-selective $k$-user interference channel assisted by RIS was first studied from a DoF perspective in \cite{bafghi2022degrees}. The authors showed that any DoF less than $k$ can be achieved with a sufficiently large number of RIS elements. Specifically, the full DoF $k$ can be achieved if the number of active RIS elements exceeds $k\left( {k - 1} \right)$. Active RIS means it can modify both phase and amplitude.
For the passive RIS, which can only adjust phase, infinite elements are required to achieve this full DoF. 
A tighter boundary $2k\left( {k - 1} \right)$ for passive RIS elements was provided in \cite{jiang2022interference}. 
However, this boundary will become increasingly loose as the strength of the direct link increases. Meanwhile, the authors of \cite{chae2022cooperative} extended this single-input single-output (SISO) network to the $k$-user MIMO network and showed that the sum rate and the sum DoF can be significantly improved if the number of active RIS elements is large enough.

\subsection{Main Contribution}
Although there have been many studies on RIS, the required number of RIS elements for the desired system performance is rarely studied. However, the system performance is significantly influenced by the number of RIS elements, as it relies on the phase adjustments of each element. This paper aims to determine the required number of RIS elements for the achievable DoF of the cellular networks, thereby providing theoretical guidance for the practical system design. The main contributions of this paper are summarized as follows: 
\begin{itemize}	
\item{We first analyze the maximum DoF through the feasibility of interference nulling equations and determine that, in the practical $G$-cell system with more single-antenna users $K$ than BS antennas $M$ per cell, the assistance of RIS is indispensable to achieve the full DoF $\min \left\{ {M,K} \right\}=M$ of each cell.
}

\item{To determine the required number of RIS elements $N$ for this full DoF,
we first propose a modified Gordon's Theorem to derive the necessary condition of $N$. Then the approximate sufficient condition of $N$ is derived by the geometric relationship between the interference null equations and the modulus-1 constraints. We prove that both of them have the same form of threshold, i.e., 
${O\left( {\sqrt L \eta } \right)}$ for ${\eta  \ge O\left( {\sqrt L } \right)}$, and ${O\left( L \right)}$ for else,
where $L{\text{ = }}GM\left( {GM - 1} \right)$ and $\eta$ is the strength ratio of the direct link to the cascaded reflective link.
}

\item{Moreover, we use the extreme value statistics to derive a more precise necessary condition. Specifically, for large $\eta $, the threshold is $L + \left( {\sqrt L  + c} \right)\eta $ ($c$ is a constant). For small $\eta $, it is the same as when $\eta=0 $, i.e., slightly larger than ${2L}$. And we show that the phase transition point of $\eta $ is approximately $\sqrt L $.}
		
\item{To verify the above conclusions, we propose a sum rate maximization problem and employ a Riemannian conjugate gradient method (RCG) to solve it. The simulation results show that our conclusions are highly consistent with the actual performances. }
\end{itemize}

\subsection{Paper Organization and Notations}
The rest of the paper is organized as follows. Section II introduces the RIS-aided multi-cell system model. Section III analyzes the available DoF of each cell. In Section IV, we derive the order of the required reflective elements for full-DoF. Section V presents a more precise condition. In Section VI, we propose a sum rate maximization problem. Simulation results are shown in Section VII. Finally, Section VIII concludes the paper.

\emph{Notations:} The scalar, vector and matrix are denoted by lowercase letter ($x$), boldface lowercase letter (${\bf{x}}$) and boldface uppercase letter (${\bf{X}}$), respectively. We use ${\left(  \cdot  \right)^*}$, ${\left(  \cdot  \right)^T}$ and ${\left(  \cdot  \right)^H}$ to denote the conjugate, transpose and conjugate transpose respectively.  ${\left(  \cdot  \right)^ + }$ , $\left|  \cdot  \right|$ represent  the pseudo-inverse and Frobenius norm respectively. The real part and image part of a complex number 
are denoted by ${\rm{re}}\left(  \cdot  \right)$ and ${\rm{im}}\left(  \cdot  \right)$ respectively. $E\left(  \cdot  \right)$ and ${\mathop{\rm var}} \left(  \cdot  \right)$ represent the expectation value and variance, respectively.
We use ${\text{ve}}{{\text{c}}_k}\left( {\mathbf{X}} \right)$ to denote the $k$-th column of ${\bf{X}}$.

\section{System Model}
\label{System_Model}
As shown in Fig. \ref{fig1}, the RIS-aided multi-cell network consists of $G$ cells and one RIS with $N$ passive reflective elements. Each cell is equipped with a BS that has $M$ antennas and supports $K$ single-antenna users. Users can transmit their messages to the BS via both direct and reflecting channels facilitated by the RIS. Since all transmission links share the same time and frequency resources, each BS receives not only the desired signals from its local users but also interference from users in other cells.

We denote ${{\mathbf{H}}_i} \in {\mathbb{C}^{N \times M}}$ as the reflecting channel between the RIS and the BS of cell $i$, ${\mathbf{h}}_{Ik}^{[i]} \in {\mathbb{C}^{N \times 1}}$ as the reflecting channel between the RIS and the $k$-th user in cell $i$, and ${\mathbf{h}}_{{B_i}k}^{[j]} \in {\mathbb{C}^{M \times 1}}$ as the direct channel between the BS of cell $i$ and the $k$-th user in cell $j$. The reflection coefficients of the RIS are represented by a diagonal matrix  ${\mathbf{R}} = diag\left( {{e^{j{\theta _1}}}, \cdots,{e^{j{\theta _N}}}} \right) \in {\mathbb{C}^{N \times N}}$, where $0 \le {\theta _n} \le 2\pi,\forall n \in 1, \cdots, N$ is the phase shift of each passive element of the RIS.
Therefore, the received signal at the BS of cell $i$ can be expressed as:

\begin{equation}
\begin{aligned}
{{\bf{y}}_i} &= \sum\limits_{g = 1}^G {\sum\limits_{k = 1}^K {\left( {{\bf{H}}_i^H{\bf{Rh}}_{Ik}^{[g]} + {\mathbf{h}}_{{B_i}k}^{[g]}} \right)x_k^{[g]}} }  + {{\bf{n}}_i}\\
&= \sum\limits_{g = 1}^G {\sum\limits_{k = 1}^K {\left( {{\bf{H}}_i^Hdiag\left( {{\bf{h}}_{Ik}^{[g]}} \right){\bf{v}} + {\mathbf{h}}_{{B_i}k}^{[g]}} \right)x_k^{[g]} + {{\bf{n}}_i}} } \\
&= \sum\limits_{g = 1}^G {\sum\limits_{k = 1}^K {\left( {{{\left( {{\bf{A}}_{ik}^{[g]}} \right)}^H}{\bf{v}} + {\mathbf{h}}_{{B_i}k}^{[g]}} \right)x_k^{[g]} + {{\bf{n}}_i}} } 
\end{aligned},
\label{eq1}
\end{equation}
where 
\begin{equation}
\begin{array}{l}
{\left( {{\mathbf{A}}_{ik}^{[g]}} \right)^H} = {\mathbf{H}}_i^Hdiag\left( {{\mathbf{h}}_{Ik}^{[g]}} \right) \in {\mathbb{C}^{M \times N}}
\end{array}
\label{eq2}
\end{equation}
denotes the cascaded reflecting channel from the $k$-th user of cell $g$ to the BS of cell $i$, ${\mathbf{v}} = {\left[ {{e^{j{\theta _1}}}, \cdots,{e^{j{\theta _N}}}} \right]^T} \in {\mathbb{C}^{N \times 1}}$ is the phase shift vector, and ${{\mathbf{n}}_i} \in {\mathbb{C}^{M \times 1}} \sim \mathcal{C}\mathcal{N}\left( {0,\sigma^2 {I_M}} \right)$ is the additive white Gaussian noise (AWGN) at the BS of cell $i$. $x_k^{[g]} \in {\mathbb{C}^{1 \times 1}}$ corresponds to the signal transmitted by the $k$-th user of cell $g$ with the transmit power of $p_k^{[g]}$, i.e., $E\left( {{{\left| {x_k^{[g]}} \right|}^2}} \right) = p_k^{[g]}$.

\begin{figure}
	\centering\includegraphics[width=7.5cm]{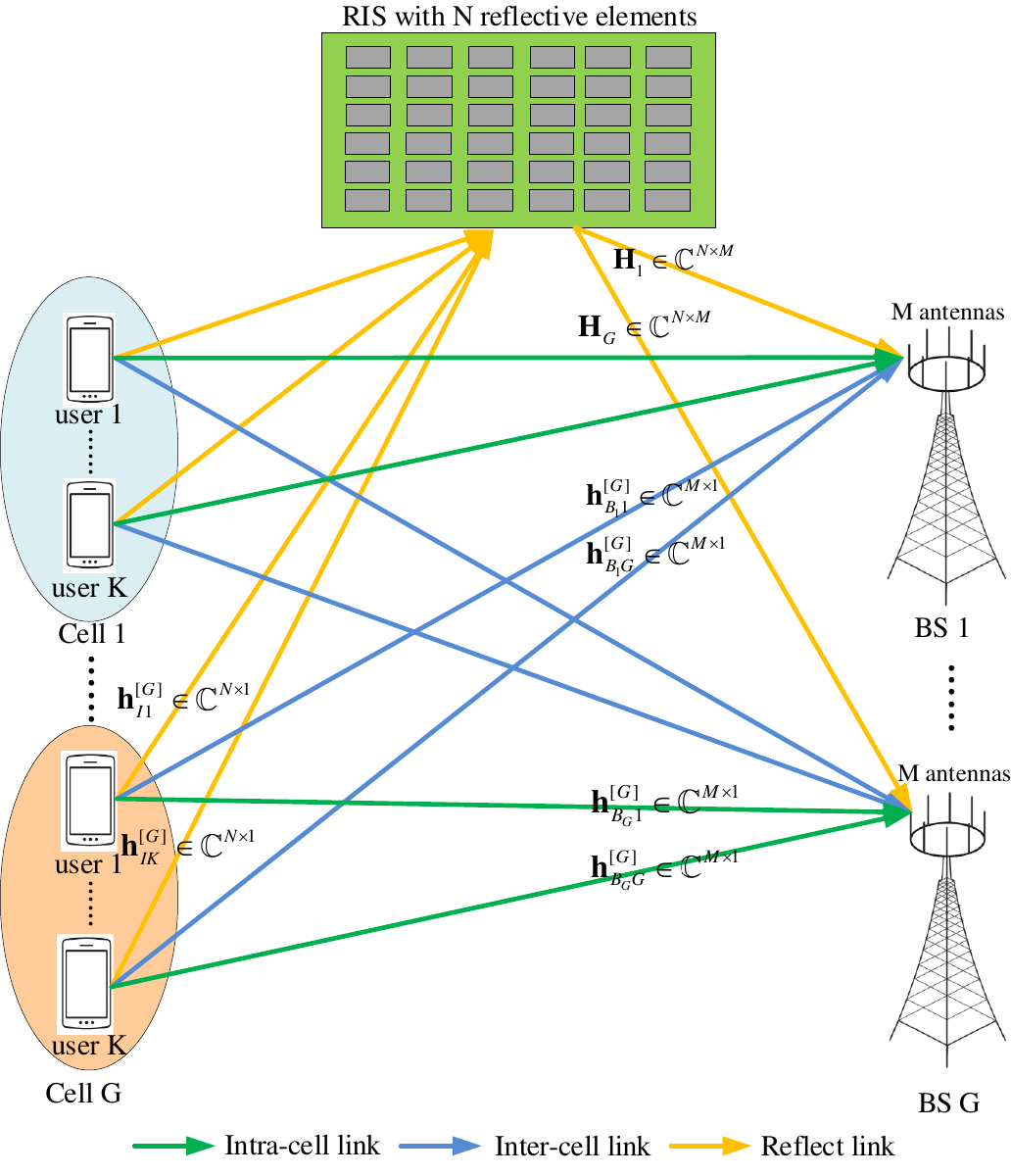}
	\caption{The RIS-aided multi-cell system}
	\label{fig1}
\end{figure}

Moreover, it has been demonstrated that the required reflective elements to eliminate interference are independent of the channel model \cite{jiang2022interference}. Therefore, without loss of generality, we assume that all channels suffer from Rayleigh fading.
Although path losses vary across different cells, the key factor affecting the required RIS elements $N$ is the significant intensity disparity between direct and cascaded reflective links. To explore this effect, we first assume uniform path losses for all users.
Specifically, each element of  ${{\mathbf{H}}_i} \in {\mathbb{C}^{N \times M}}$ is i.i.d.$ \sim \mathcal{C}\mathcal{N}\left( {0,\sigma _1^2} \right)$,
${\bf{h}}_{Ik}^{[i]} \sim {\cal C}{\cal N}\left( {0,\sigma _2^2{{\bf{I}}_N}} \right)$, and ${\mathbf{h}}_{{B_i}k}^{[j]} \sim {\cal C}{\cal N}\left( {0,\sigma _3^2{{\bf{I}}_M}} \right)$.
If $N$ and $L$ are large enough, we approximate ${\mathbf{A}} \in {\mathbb{C}^{N \times L}}$ as Gaussian distribution, and the variance of each element is $\sigma _4^2$ (where ${\sigma _4} = {\sigma _1}{\sigma _2}$).
The scenario with different path losses for each user is simulated in Section \ref{different_path}, and the results are consistent with those obtained from the previous assumptions.

In this paper, we also assume that the RIS is configured by a centralized controller, and the channel state information (CSI) of all links is perfectly known at the centralized controller and the BS \cite{chen2023channel,huang2023integrating,guo2023joint}.

\section{Achievable DoF of Each Cell}
To achieve the full DoF or the interference-free DoF of each cell, interference must be completely eliminated. In this section, we analyze the achievable DoF of each cell through the feasibility of interference nulling equations.

In the practical $G$-cell system, due to the size and cost limitations, the number of users $K$ in each cell is usually larger than the number of antennas $M$ at each BS, i.e., $K \ge M$. 
\begin{theorem}
If ${\rm{M}} < KG$, the full DoF $\min \left\{ {M,K} \right\}$ \cite{cover1999elements} of each cell cannot be achieved by antenna collaboration alone, and the assistance of RIS is indispensable.
\label{the1}
\end{theorem}

\begin{proof}
Please see Appendix \ref{appA}. 
\end{proof}

Given that $M$ is much less than $KG$ in the practical multi-cell system, the contribution of antenna collaboration to achieving the full DoF is negligible compared to that of the RIS. Therefore, this paper focuses solely on enhancing the DoF through the RIS. 

Since the full DoF of each cell is $\min \left\{ {M,K} \right\}$, we can assume without loss of DoF that $M = K$. In other words, if the number of users $\tilde K$ in each cell is larger than the number of antennas $M$ at each BS, we can randomly set $\tilde K - M$ user's transmit power to zero in each time slot \cite{cadambe2008interference,cadambe2009interference,5502363}. The users' messages are then decoded from distinct rows of the received signal at the BS. Specifically, the message of the $k$-th user in cell $i$ is decoded from the $k$-th row of ${{{\bf{ y}}}_i}$ as
\begin{equation}
\begin{aligned}
{{{\bf{ y}}}_{ik}}{\rm{ = }}&\left( {\left( {{\bf{A}}_{ik}^{[i]}} \right)_k^H{\bf{v}} + {{\left( {{\mathbf{h}}_{{B_i},k}^{[i]}} \right)}_k}} \right)x_k^{[i]}\\
&+ \underbrace {\sum\limits_{j \ne k}^K {\left( {\left( {{\bf{A}}_{ij}^{[i]}} \right)_k^H{\bf{v}} + {{\left( {{\mathbf{h}}_{{B_i},j}^{[i]}} \right)}_k}} \right)x_j^{[i]}} }_{{\rm{\text{intra-cell interference}}}}\\
&+ \underbrace {\sum\limits_{g \ne i}^G {\sum\limits_{j = 1}^K {\left( {\left( {{\bf{A}}_{ij}^{[g]}} \right)_k^H{\bf{v}} + {{\left( {{\mathbf{h}}_{{B_i},j}^{[g]}} \right)}_k}} \right)x_j^{[g]}} } }_{{\rm{\text{inter-cell interference}}}} + {{\bf{n}}_{ik}}
\end{aligned},
\label{eq22} 
\end{equation}
where $\left( {{\bf{A}}_{ij}^{[g]}} \right)_k^H$ and ${\left( {{\mathbf{h}}_{{B_i},j}^{[g]}} \right)_k}$ are the $k$-th row of ${\left( {{\bf{A}}_{ij}^{[g]}} \right)^H}$ and ${{\mathbf{h}}_{{B_i},j}^{[g]}}$, respectively.

To achieve the full-DoF, the following equations must hold simultaneously for all $i,{\rm{g}} \in \left[ {1,G} \right]$ and $k \in \left[ {1,K} \right]$.
\begin{equation}
\left\{ {\begin{aligned}
	&{\left( {{\bf{A}}_{ik}^{[i]}} \right)_k^H{\bf{v}} + {{\left( {{\mathbf{h}}_{{B_i},k}^{[i]}} \right)}_k} \ne 0}\\
	&{\left( {{\bf{A}}_{ij}^{[i]}} \right)_k^H{\bf{v}} + {{\left( {{\mathbf{h}}_{{B_i},j}^{[i]}} \right)}_k} = 0,{\rm{ }}j \ne k}\\
	&{\left( {{\bf{A}}_{ij}^{[g]}} \right)_k^H{\bf{v}} + {{\left( {{\mathbf{h}}_{{B_i},j}^{[g]}} \right)}_k} = 0,{\rm{ g}} \ne i,j \in \left[ {1,K} \right]}
	\end{aligned}} \right.
\label{eq23}  
\end{equation}
It is possible due to the randomness of channel \cite{jiang2022interference}. Therefore, the conditions for the full-DoF can be transformed into a feasible problem with $L{\text{ = }}GM\left( {GM - 1} \right)$ zero equations as

\begin{equation}
\begin{gathered}
{\text{find }}{\mathbf{v}} \hfill \\
{\text{s}}{\text{.t}}{\text{. }}{\mathbf{v}} \in {S_1} \cap {S_2} \hfill \\ 
\end{gathered}
\label{eq25}  
\end{equation}
where
\begin{equation}
\left\{ {\begin{aligned}
	&{{S_1}:{{\mathbf{A}}^H}{\mathbf{v}}{\text{ + }}{\mathbf{b}} = 0} \\ 
	&{{S_2}:\left| {{v_i}} \right| = 1,i \in \left[ {1,N} \right]} 
	\end{aligned}} \right.,
\label{eq26}  
\end{equation}
and
\begin{equation}
\begin{array}{l}
{\mathbf{A}} = \left[ {{\text{ve}}{{\text{c}}_k}\left( {{\mathbf{A}}_{ij}^{[i]}} \right), \cdots ,{\text{ve}}{{\text{c}}_k}\left( {{\mathbf{A}}_{ij}^{[g]}} \right)} \right] \in {\mathbb{C}^{N \times L}},
\end{array}  
\end{equation}
\begin{equation}
\begin{array}{l}
{\mathbf{b}} = {\left[ {{{\left( {{\mathbf{h}}_{{B_i},j}^{[i]}} \right)}_k}, \cdots ,{{\left( {{\mathbf{h}}_{{B_i},j}^{[g]}} \right)}_k}} \right]^T} \in {\mathbb{C}^{L \times 1}}.
\end{array}
\end{equation}
Problem \eqref{eq25} can be solved if the number of RIS elements $N$ is large enough \cite{jiang2022interference}. But excessive $N$ will increase the implementation costs and insufficient $N$ cannot guarantee the solvability of \eqref{eq25}. 
The required $N$ for the solvability of \eqref{eq25} is discussed in the next two sections.

\section{the Order of $N$ for the Full-DoF}
In high-dimensional space, although ${\mathbf{v}} \in {\mathbb{C}^{N \times 1}}$ has complex geometric properties such as magnitude and angle, these complexities will be averaged out when it is randomly projected onto a lower-dimensional space. As a result, the magnitude of $\left| {{{\mathbf{A}}^H}{\mathbf{v}}} \right|$ tends to be more stable and concentrated around a narrow range, which is known as the norm concentration phenomenon in high-dimensional statistical \cite{vershynin2018high}. If $\left| {\mathbf{b}} \right|$ does not lie within this range, then no matter how ${\mathbf{v}}$ is adjusted, the equation of \eqref{eq26} will almost 
have no solution. Therefore, in this section, we derive the necessary condition of $N$ for full-DoF based on this norm concentration phenomenon.

Furthermore, ${S_2}$ is a symmetric subset of the sphere concerning the center, and ${S_1}$ represents the solution space of the equation. If the sphere's radius exceeds the distance from its center to ${S_1}$, the sphere and ${S_1}$ will intersect.
Therefore, this geometric relationship between the two sets is also utilized in this section to derive the approximate sufficient condition of $N$ for achieving the full-DoF.

\subsection{the Necessary Condition}
\label{Necessary_A}
In this subsection, we propose a modified Gordon's Theorem to derive the necessary $N$ for full DoF based on the norm concentration phenomenon in high dimensional probability.

Specifically, for ${{\bf{v}}_{{S_1}}} = \left\{ {\left. {\bf{v}} \right|{\bf{v}} \in {S_1}} \right\}$, we must have $\left| {{{\bf{A}}^H}{{\bf{v}}_{{S_1}}} + {\bf{b}}} \right| = 0$. If $\min \left| {{{\bf{A}}^H}{{\mathbf{v}}_{{S_2}}}{\text{ + }}{\mathbf{b}}} \right| > 0$ for all ${{\bf{v}}_{{S_2}}} = \left\{ {\left. {\bf{v}} \right|{\bf{v}} \in {S_2}} \right\}$, then there is no intersection between ${{S_1}}$ and ${{S_2}}$, i.e., ${S_1} \cap {S_2} = \phi $.
And because
\begin{equation}
\left| {\left| {{{\bf{A}}^H}{{\bf{v}}_{{S_2}}}} \right| - \left| {\bf{b}} \right|} \right| \le \left| {{{\bf{A}}^H}{{\bf{v}}_{{S_2}}}{\rm{ + }}{\bf{b}}} \right| \le \left| {{{\bf{A}}^H}{{\bf{v}}_{{S_2}}}} \right| + \left| {\bf{b}} \right|
\label{eq28},    
\end{equation}
$\min \left| {{{\bf{A}}^H}{{\mathbf{v}}_{{S_2}}}{\text{ + }}{\mathbf{b}}} \right| > 0$ is equivalent to
$\min \left| {{{\bf{A}}^H}{{\mathbf{v}}_{{S_2}}}} \right| > \left| {\mathbf{b}} \right|$ or $\max \left| {{{\bf{A}}^H}{{\mathbf{v}}_{{S_2}}}} \right| < \left| {\mathbf{b}} \right|$.

According to the norm concentration phenomenon in high dimension, $\left| {{{\bf{A}}^H}{{\bf{v}}_{{S_2}}}} \right|$ is concentrated around a certain expected value within a narrow range. Gordon's Theorem \cite{bandeira2015ten} indicates that if ${{\bf{A}}^H}$ is a random matrix with independent $\mathcal{N}\left( {0,1} \right)$ entries and ${{\bf{v}}_{{S_2}}}$ is a unit vector, then $\left| {{{\bf{A}}^H}{{\bf{v}}_{{S_2}}}} \right|$ approximately has the same expected value as $\left| {\mathbf{g}} \right|$, where ${\mathbf{g}}$ is a $L$-dimensional Gaussian random vector.
Moreover, due to the diversity of ${{\bf{v}}_{{S_2}}}$, their specific projection lengths will fluctuate around this expected value.
This fluctuation can be quantified by the Gaussian width $W\left( S \right)$ \cite{vershynin2018high} of the set $S$ composed of ${{\bf{v}}_{{S_2}}}$, as the Gaussian width $W\left( S \right)$ represents the expected maximum projection length of these vectors of $S$ in a Gaussian random direction.
However, in Section \ref{System_Model}, the entries of ${{\bf{A}}^H}$ and ${\bf{b}}$ are assumed to independently follow $\mathcal{C}\mathcal{N}\left( {0,{\sigma _4^2}} \right)$ and $\mathcal{C}\mathcal{N}\left( {0,{\sigma _3^2}} \right)$, respectively. ${{\bf{v}}_{{S_2}}}$ is also not an unit vector. Thus, we proposed a modified Gordon's Theorem as follows to determine the value of $\left| {{{\bf{A}}^H}{{\bf{v}}_{{S_2}}}} \right|$.

\begin{theorem}
Let ${\mathbf{G}} \in {\mathbb{C}^{L \times N}}$ be a random matrix with independent $\mathcal{C}\mathcal{N}\left( {0,{\sigma ^2}} \right)$ elements and the set ${S} = \left\{ {\left. {\bf{x}} \right|\left| {{x_i}} \right| = 1,i \in \left[ {1,N} \right]} \right\}$. Then
\begin{equation}
\left\{ {\begin{array}{*{20}{c}}
	{E\left( {\mathop {\min }\limits_{{\bf{x}} \in S} \left| {{\bf{Gx}}} \right|} \right) \ge \sigma \sqrt N \left( {\sqrt L  - \frac{{\sqrt \pi  }}{2}\sqrt N } \right)}\\
	{E\left( {\mathop {\max }\limits_{{\bf{x}} \in S} \left| {{\bf{Gx}}} \right|} \right) \le \sigma \sqrt N \left( {\sqrt L  + \frac{{\sqrt \pi  }}{2}\sqrt N } \right)}
	\end{array}} \right.
\label{eq29}   
\end{equation}
\label{the2}
\end{theorem}

\begin{proof}
Please see Appendix \ref{appB}.
\end{proof}

\begin{theorem}	
If the random vector ${\mathbf{x}} \sim \mathcal{C}\mathcal{N}\left( {0,{\sigma ^2}{{\mathbf{I}}_L}} \right)$, then the expectation value of its norm is $E\left( {\left| {\mathbf{x}} \right|} \right) = \sigma \sqrt L $ \cite{vershynin2018high}.
\label{the3}
\end{theorem}

According to the Theorem \ref{the2}, Theorem \ref{the3} and the norm concentration phenomenon, $\left| {{{\bf{A}}^H}{{\bf{v}}_{{S_2}}}} \right|$ and $\left| {\mathbf{b}} \right|$ can be respectively approximated as
\begin{equation}
\begin{aligned}
{\sigma _4}\sqrt N \left( {\sqrt L  - \frac{{\sqrt \pi  }}{2}\sqrt N } \right) \le \left| {{{\bf{A}}^H}{{\bf{v}}_{{S_2}}}} \right|\\
\le {\sigma _4}\sqrt N \left( {\sqrt L  + \frac{{\sqrt \pi  }}{2}\sqrt N } \right)
\end{aligned}
\label{eq36}  
\end{equation}
and 
\begin{equation}
\left| {\mathbf{b}} \right| = {\sigma _3}\sqrt {L}.
\label{eq37} 
\end{equation}

Firstly, for the case of $\max \left| {{{\bf{A}}^H}{{\bf{v}}_{{S_2}}}} \right| < \left| {\bf{b}} \right|$ we have
\begin{equation}
{\sigma _4}\sqrt N \left( {\sqrt L  + \frac{{\sqrt \pi  }}{2}\sqrt N } \right) < {\sigma _3}\sqrt L 
\label{eq38},  
\end{equation}
the solution of \eqref{eq38} can be calculated as
\begin{equation}
\begin{array}{l}
N < {\left( {\frac{{ - \sqrt L  + \sqrt {L + 2\sqrt \pi  \sqrt L \eta } }}{{\sqrt \pi  }}} \right)^2}
\end{array}
\label{eq39},  
\end{equation}
where $\eta  = \frac{{{\sigma _3}}}{{{\sigma _4}}}$.
 
Then, for the case of $\min \left| {{{\bf{A}}^H}{{\bf{v}}_{{S_2}}}} \right| > \left| {\bf{b}} \right|$. It can be equivalently expressed as 
\begin{equation}
\begin{array}{l}
{\sigma _4}\sqrt N \left( {\sqrt L  - \frac{{\sqrt \pi  }}{2}\sqrt N } \right) > {\sigma _3}\sqrt L
\end{array}
\label{eq40}.  
\end{equation}
Thus we can get 
\begin{equation}
\begin{array}{l}
N < {\left( {\frac{{\sqrt L  + \sqrt {L - 2\sqrt \pi  \sqrt L \eta } }}{{\sqrt \pi  }}} \right)^2}
\end{array}
\label{eq41},  
\end{equation}
where $L \ge 4\pi {\eta ^2}$, and $\eta  = \frac{{{\sigma _3}}}{{{\sigma _4}}}$. 
However, this threshold will decrease as $\eta $ increases. Which means the larger $\left| {\mathbf{b}} \right|$, the fewer $N$ are required. This contradicts the fact that more columns of ${{{\mathbf{A}}^H}}$ are needed to offset the amplitude disparity with large ${\mathbf{b}}$, as ${\mathbf{v}}$ can only perform rotations on ${{{\mathbf{A}}^H}}$'s columns.

Therefore, based on the above analysis, only the first case is consistent with reality. Moreover, as the threshold in \eqref{eq39} is derived within the high-dimensional statistical framework, it can be rewritten in the form of the order of $N$ as
\begin{equation}
N < \left\{ {\begin{aligned}
&{O\left( {\sqrt L \eta } \right),\;{\rm{if}}\;\eta  \ge O\left( {\sqrt L } \right)}\\
&{O\left( L \right),\;{\rm{else}}}
\end{aligned}} \right.
\label{eq42}.  
\end{equation}

The necessary condition means that if the number of RIS elements is less than this threshold, the full-DoF is impossible (i.e., ${S_1} \cap {S_2} = \phi $). Thus, the necessary condition can also be regarded as a lower bound of the number of RIS elements.

\subsection{the Sufficient Condition}
\label{Sufficient_B}
In this subsection, the geometric relationship between ${S_1}$ and ${S_2}$ is employed to derive the sufficient condition of $N$ for achieving full DoF. 

Specifically, since ${S_2} = \left\{ {\left. {\mathbf{v}} \right|\left| {{v_i}} \right| = 1,i \in \left[ {1,N} \right]} \right\}$ (i.e., ${{\bf{v}}^H}{\bf{v}} = N$), ${S_2}$ is a symmetric subset of the sphere with a radius of $\sqrt N $. And ${S_1}$ can be transformed as
\begin{equation}
{S_1}:{{\bf{A}}^H}{\bf{v}}{\rm{  +  }}{\bf{b}} = 0 \Rightarrow {\bf{v}} =  - {\left( {{{\bf{A}}^H}} \right)^ + }{\bf{b}} + {\rm{null}}\left( {{{\bf{A}}^H}} \right)
\label{add1}.   
\end{equation}
Therefore, the geometric relationship between ${S_1}$ and ${S_2}$ can be represented by Fig. \ref{geometry}.
\begin{figure}[H] 
	\centering
	\centering\includegraphics[width=7cm]{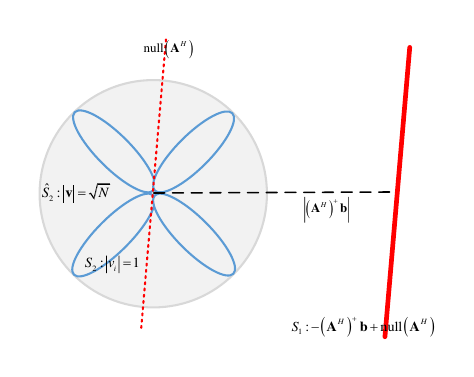}
	\caption{the geometric relationship between ${S_1}$ and ${S_2}$}
	\label{geometry}
\end{figure}
Where the blue line denotes the set ${S_2}$, and the gray sphere denotes the set ${{\hat S}_2} = \left\{ {\left. {\bf{v}} \right|{{\bf{v}}^H}{\bf{v}} = N} \right\}$. The set ${S_1}$ can be obtained by shifting the null space of ${\bf{A}}^H$ with a vector ${\left( {{{\mathbf{A}}^H}} \right)^ + }{\mathbf{b}}$. Consequently, the distance between ${S_1}$ and ${\rm{null}}\left( {{{\bf{A}}^H}} \right)$ can be characterized by $\left| {{{\left( {{{\bf{A}}^H}} \right)}^ + }{\bf{b}}} \right|$ \cite{larsen2021many}.
Therefore, to ensure ${S_1} \cap {{\hat S}_2} \ne \phi$, the intuitive approach is to set the sphere's radius greater than $\left| {{{\left( {{{\bf{A}}^H}} \right)}^ + }{\bf{b}}} \right|$, i.e., 
\begin{equation}
\begin{array}{l}
\sqrt N  \ge \left| {{{\left( {{{\bf{A}}^H}} \right)}^ + }{\bf{b}}} \right|
\end{array}
\label{add44}.  
\end{equation}
Due to the symmetry of the sphere and the set ${S_2}$, we will use the relaxed ${S_2}$ (i.e., ${{\hat S}_2}$) to determine the approximate sufficient condition.

\begin{theorem}	
	Let ${\mathbf{G}} \in {\mathbb{C}^{L \times N}},\left( {L{\rm{ + }}1 < N} \right)$ be a random matrix with independent ${\cal C}{\cal N}\left( {0,\sigma _1^2} \right)$ elements, ${\mathbf{x}} \sim \mathcal{C}\mathcal{N}\left( {0,\sigma _2^2{{\mathbf{I}}_L}} \right)$ and  $\rho  = \frac{{{\sigma _2}}}{{{\sigma _1}}}$, then the expectation norm of ${{{\mathbf{G}}^ + }{\mathbf{x}}}$ is
	\begin{equation}
	\begin{array}{l}
	E\left( {\left| {{{\bf{G}}^ + }{\bf{x}}} \right|} \right) \le \sqrt {\frac{L}{{N - L - 1}}} \rho
	\end{array}  
	\label{add2}.  
	\end{equation}
	\label{the4}
\end{theorem}

\begin{proof}
	Please see Appendix \ref{appD}.	
\end{proof}	

According to the Theorem \ref{the4} and the concentration phenomenon \cite{vershynin2018high}, inequality \eqref{add44} can be approximately as
\begin{equation}
\begin{array}{l}
\sqrt N  \ge \sqrt {\frac{L}{{N - L - 1}}} \eta. 
\end{array}   
\end{equation}
Then we can get 
\begin{equation}
\begin{array}{l}
N \ge \frac{{L + 1 + \sqrt {{{\left( {L + 1} \right)}^2} + 4L{\eta ^2}} }}{2} \buildrel \Delta \over = {N_1}
\end{array}
\label{eq32_1}.   
\end{equation}
The simulation results demonstrate that ${N_1}$ is almost consistent with the condition that ${S_1}$ and ${{\hat S}_2}$ intersect with a $50\%$ probability.
This $50\%$ phenomenon can also be found in \cite{amelunxen2014living}. 
Moreover, the condition for $50\%$ intersection probability is symmetrically positioned between the conditions for high (e.g., $99\%$) and low  (e.g., $1\%$) probabilities. Thus, the sufficient condition for ${S_1} \cap {{\hat S}_2} \ne \phi $ can be approximated as $2{N_1}$ minus the necessary condition.
These phenomenons are also illustrated in Fig. \ref{int_cir} of the Section \ref{geometry1}.

Similar to the analysis in Section \ref{Necessary_A}, the necessary condition for ${S_1} \cap {{\hat S}_2} \ne \phi $ can be derived by $\max \left| {{{\bf{A}}^H}{{\bf{v}}_{{{\hat S}_2}}}} \right| < \left| {\bf{b}} \right|$ and Theorem \ref{the5} as
\begin{equation}
\begin{array}{l}
N < {\left( {\frac{{ - \sqrt L  + \sqrt {L + 4\sqrt L \eta } }}{2}} \right)^2} \buildrel \Delta \over = {N_2}.
\label{eq34_1}   
\end{array}
\end{equation}

\begin{theorem}
	Let ${\mathbf{G}} \in {\mathbb{C}^{L \times N}}$ be a random matrix with independent $\mathcal{C}\mathcal{N}\left( {0,{\sigma ^2}} \right)$ elements and the set $S = \left\{ {\left. {\bf{x}} \right|{{\bf{x}}^H}{\bf{x}} = N} \right\}$. Then
	\begin{equation}
	\left\{ {\begin{array}{*{20}{c}}
		{E\left( {\mathop {\min }\limits_{{\bf{x}} \in S} \left| {{\bf{Gx}}} \right|} \right) \ge \sigma \sqrt N \left( {\sqrt L  - \sqrt N } \right)}\\
		{E\left( {\mathop {\max }\limits_{{\bf{x}} \in S} \left| {{\bf{Gx}}} \right|} \right) \le \sigma \sqrt N \left( {\sqrt L  + \sqrt N } \right)}
		\end{array}} \right.
	\label{eq33_1}    
	\end{equation}
	\label{the5}
\end{theorem}

\begin{proof}
	Please see Appendix \ref{appE}.
\end{proof}

Therefore, the sufficient condition for ${S_1} \cap {{\hat S}_2} \ne \phi $ can be calculated by \eqref{eq32_1} and \eqref{eq34_1} as
\begin{equation}
\begin{aligned}
N \ge& 2{N_1} - {N_2}\\
 =& L + 1 + \sqrt {{{\left( {L + 1} \right)}^2} + 4L{\eta ^2}} \\
 &- {\left( {\frac{{ - \sqrt L  + \sqrt {L + 4\sqrt L \eta } }}{2}} \right)^2}
\end{aligned}
\label{eq35_1}.  
\end{equation}

However, since this result is obtained within the high dimensional statistical framework and based on the relaxed ${S_2}$, it can only be regarded as an approximate sufficient condition, 
and can be rewritten as
\begin{equation}
N \ge \left\{ {\begin{aligned}
&{O\left( {\sqrt L \eta } \right),\;{\rm{if}}\;\eta  \ge O\left( {\sqrt L } \right)}\\
&{O\left( L \right),\;{\rm{else}}}
\end{aligned}} \right.
\label{add4}.  
\end{equation}

The sufficient condition means that if the number of RIS elements is larger than this threshold, the full DoF can almost be guaranteed. 
Thus, this sufficient condition can also be approximated as an upper bound of the RIS elements.

\section{More Precise Necessary Condition of $N$}
\label{Precise}
The conditions of $N$ for full DoF have been derived from high-dimensional probability and geometric perspective in Section IV. It reveals which factors are associated with $N$ and how they affect $N$ asymptotically. 

Since each element of ${\mathbf{A}}$ and ${\mathbf{b}}$ follows Gaussian distribution, their linear combinations also follow Gaussian distribution. If $N$ and $L$ are large, the magnitude's probability density function (PDF) of them is tall and narrow. Therefore,
if the right tail of the magnitude's PDF for the linear combination of ${\mathbf{A}}$ is smaller than the left tail of the corresponding ${\mathbf{b}}$'s PDF, it is almost impossible to find a ${\mathbf{v}}$ such that ${{{\mathbf{A}}^H}{\mathbf{v}}} = -{\mathbf{b}}$.
In this section, a more precise necessary condition is derived based on the above properties.

Specifically, the set $S_1$ can be expanded as
\begin{equation}
\begin{aligned}
{a_{11}}{v_1} +  \cdots  &+ {a_{1N}}{v_N} =  - {b_1}\\
 &\vdots \\
{a_{L1}}{v_1} +  \cdots  &+ {a_{LN}}{v_N} =  - {b_L}
\end{aligned}
\label{eq43},   
\end{equation}
thus, if all the $L$ equations hold simultaneously, there must be
\begin{equation}
\begin{array}{l}
\left| {\sum\limits_i^L {{b_i}} } \right| = \left| {\sum\limits_i^L {{a_{i1}}} {v_1} +  \cdots  + \sum\limits_i^L {{a_{iN}}} {v_N}} \right|
\end{array}
\label{eq44}.  
\end{equation}
Since our goal is to find a ${\mathbf{v}} \in {S_1} \cap {S_2}$, after combining ${\bf{v}} \in {S_2}$ with \eqref{eq44}, we can get
\begin{equation}
\begin{array}{l}
{\left| {\sum\limits_i^L {{a_{i1}}} {v_1} +  \cdots  + \sum\limits_i^L {{a_{iN}}} {v_N}} \right| \le \sum\limits_j^N {\left| {\sum\limits_i^L {{a_{ij}}} {v_j}} \right|}  \le \sum\limits_j^N {\left| {\sum\limits_i^L {{a_{ij}}} } \right|} }
\end{array}
\label{eq45}. 
\end{equation}
And \eqref{eq45} can be rewritten as
\begin{equation}
\begin{array}{l}
	\left| {\sum\limits_i^L {{a_{i1}}} {v_1} +  \cdots  + \sum\limits_i^L {{a_{iN}}} {v_N}} \right| + \alpha  = \sum\limits_j^N {\left| {\sum\limits_i^L {{a_{ij}}} } \right|} ,\alpha  \ge 0 
	\end{array}
	\label{eq47}  
\end{equation}
and $E\left( {\sum\limits_j^N {\left| {\sum\limits_i^L {{a_{ij}}} } \right|} } \right) = N\sqrt L {\sigma _4}\frac{{\sqrt \pi  }}{2}$.

If we use the equation $\left| {\sum\limits_i^L {{b_i}} } \right| = \sum\limits_j^N {\left| {\sum\limits_i^L {{a_{ij}}} } \right|} $ to determine the minimal required $N$, it is equivalent to
\begin{equation}
\begin{array}{l}
	\left| {\sum\limits_i^L {{b_i}} } \right| = \left| {\sum\limits_i^L {{a_{i1}}} {v_1} +  \cdots  + \sum\limits_i^L {{a_{iN}}} {v_N}} \right| + \alpha
	\end{array}
	\label{eq51_1}. 
\end{equation} 
Obviously, the number of variables $N$ calculated by \eqref{eq51_1} is less than the actual value obtained by the equation \eqref{eq44}. 

To reduce the error on $N$, we can also apply the triangle inequality on b, i.e., $\left| {\sum\limits_i^L {{b_i}} } \right| \le \sum\limits_i^L {\left| {{b_i}} \right|} $. It can be expressed as
\begin{equation}
\begin{array}{l}
\sum\limits_i^L {\left| {{b_i}} \right|}  = \left| {\sum\limits_i^L {{b_i}} } \right| + \beta ,\beta  \ge 0
\end{array}
\label{eq46}.  
\end{equation}

When the $L$ equations with $N$ variables in \eqref{eq43} holds simultaneously, there also must be
\begin{equation}
\left\{ \begin{array}{l}
\sum\limits_i^L {\left| {{b_i}} \right|}  = \sum\limits_i^L {\left| {\sum\limits_j^N {{a_{ij}}} {v_j}} \right|} \\
L \le N
\end{array} \right.
\label{eq48}.  
\end{equation} 
And because ${{a_{ij}}}$ and ${{b_i}}$ are random variables, ${{v_j}}$ can also be approximated as a random variable but with modulus-1. Then
\begin{equation}
\begin{array}{l}
E\left( {\sum\limits_i^L {\left| {{b_i}} \right|} } \right) = E\left( {\sum\limits_i^L {\left| {\sum\limits_j^N {{a_{ij}}} {v_j}} \right|} } \right)\\
\approx LE\left( {\left| {\sum\limits_j^N {{a_{ij}}} } \right|} \right) = L\sqrt N {\sigma _4}\frac{{\sqrt \pi  }}{2}
\end{array}
\label{eq49}.  
\end{equation}
Due to $L \le N$, it is obviously that $L\sqrt N {\sigma _4}\frac{{\sqrt \pi  }}{2} \le N\sqrt L {\sigma _4}\frac{{\sqrt \pi  }}{2}$, thus we have
\begin{equation}
\begin{array}{l}
E\left( {\sum\limits_i^L {\left| {{b_i}} \right|} } \right) \le E\left( {\sum\limits_j^N {\left| {\sum\limits_i^L {{a_{ij}}} } \right|} } \right)
\end{array}
\label{eq50}.  
\end{equation}
Combining \eqref{eq44}, \eqref{eq47}, \eqref{eq46} and \eqref{eq50}, we can get $E\left( \alpha  \right) \ge E\left( \beta  \right)$.
Therefore, if we use the equation 
\begin{equation}
    \sum\limits_i^L {\left| {{b_i}} \right|}  = \sum\limits_j^N {\left| {\sum\limits_i^L {{a_{ij}}} } \right|}
    \label{eq50_1}
\end{equation}
to determine the minimal required $N$, it is equivalent to
\begin{equation}
\begin{array}{l}
\left| {\sum\limits_i^L {{b_i}} } \right| = \left| {\sum\limits_i^L {{a_{i1}}} {v_1} +  \cdots  + \sum\limits_i^L {{a_{iN}}} {v_N}} \right| + \alpha  - \beta 
\end{array}
\label{eq51},  
\end{equation}
which is closer to \eqref{eq44} than  \eqref{eq51_1}.

Therefore, the necessary condition for the problem \eqref{eq25} can be derived by the extreme
value statistics of \eqref{eq50_1} as
\begin{equation}
\begin{array}{l}
E\left( {\sum\limits_i^L {\left| {{b_i}} \right|} } \right) - {c_1}\sqrt {{\mathop{\rm var}} \left( {\sum\limits_i^L {\left| {{b_i}} \right|} } \right)}  > \\
E\left( {\sum\limits_j^N {\left| {\sum\limits_i^L {{a_{ij}}} } \right|} } \right)  + {c_2}\sqrt {{\mathop{\rm var}} \left( {\sum\limits_j^N {\left| {\sum\limits_i^L {{a_{ij}}} } \right|} } \right)} 
\end{array}
\label{eq53},  
\end{equation}
where ${a_{ij}} \sim {\cal C}{\cal N}\left( {0,\sigma _4^2} \right)$ and ${{b}_i} \sim {\cal C}{\cal N}\left( {0,\sigma _3^2} \right)$, ${c_1}$ and ${c_2}$ are finite constants greater than zero.
Specifically, if the left tail of the PDF of $\sum\limits_i^L {\left| {{b_i}} \right|}$ is much larger than the right tail of the PDF of $\sum\limits_j^N {\left| {\sum\limits_i^L {{a_{ij}}} } \right|}$, then it is almost impossible to find a ${\mathbf{v}}$ that satisfies the equations of \eqref{eq43}.

The means and variances in \eqref{eq53} can be calculated as follows
\begin{equation}
\begin{array}{l}
E\left( {\sum\limits_i^L {\left| {{b_i}} \right|} } \right) = LE\left( {\left| {{b_i}} \right|} \right) = L\frac{{\sqrt \pi  }}{2}{\sigma _3}
\end{array}
\label{eq54},  
\end{equation}
\begin{equation}
\begin{array}{l}
{\mathop{\rm var}} \left( {\sum\limits_i^L {\left| {{b_i}} \right|} } \right){\rm{ = }}L{\mathop{\rm var}} \left( {\left| {{b_i}} \right|} \right){\rm{ = }}L\left( {1 - \frac{\pi }{4}} \right)\sigma _3^2
\end{array}
\label{eq55},   
\end{equation}
\begin{equation}
\begin{array}{l}
E\left( {\sum\limits_j^N {\left| {\sum\limits_i^L {{a_{ij}}} } \right|} } \right) = NE\left( {\left| {\sum\limits_i^L {{a_{ij}}} } \right|} \right) = N\frac{{\sqrt \pi  }}{2}\sqrt {L\sigma _4^2} 
\end{array}
\label{eq56},   
\end{equation}
\begin{equation}
\begin{array}{l}
{\mathop{\rm var}} \left( {\sum\limits_j^N {\left| {\sum\limits_i^L {{a_{ij}}} } \right|} } \right) = N{\mathop{\rm var}} \left( {\left| {\sum\limits_i^L {{a_{ij}}} } \right|} \right) = N\left( {1 - \frac{\pi }{4}} \right)L\sigma _4^2
\end{array}
\label{eq57}.   
\end{equation}
Substituting these means and variances into \eqref{eq53}, we have 
\begin{equation}
\begin{array}{l}
L\frac{{\sqrt \pi  }}{2}{\sigma _3} - {c_1}\sqrt {L\left( {1 - \frac{\pi }{4}} \right)\sigma _3^2}  > \\
N\frac{{\sqrt \pi  }}{2}\sqrt {L\sigma _4^2}  + {c_2}\sqrt {N\left( {1 - \frac{\pi }{4}} \right)L\sigma _4^2}  
\end{array}
\label{eq58}.  
\end{equation}
The solution to the inequality \eqref{eq58} with respect to $N$ is
\begin{equation}
\begin{array}{l}
N < \frac{{c_2^2\left( {\frac{4}{\pi } - 1} \right)}}{2} + \left( {\sqrt L  - {c_1}\sqrt {\frac{4}{\pi } - 1} } \right)\eta \\
- \frac{{{c_2}\sqrt {\frac{4}{\pi } - 1}  \cdot \sqrt {c_2^2\left( {\frac{4}{\pi } - 1} \right) + 4\left( {\sqrt L  - {c_1}\sqrt {\frac{4}{\pi } - 1} } \right)\eta } }}{2}
\end{array}
\label{eq59}.   
\end{equation}
If $\eta$ is large, then \eqref{eq59} can be approximated as
\begin{equation}
\begin{array}{l}
N < \left( {\sqrt L  - {c_1}\sqrt {\frac{4}{\pi } - 1} } \right)\eta {\rm{ = }}\left( {\sqrt L  + c} \right)\eta 
\end{array}
\label{eq60}.  
\end{equation}

However, since this result is derived from an approximate form of \eqref{eq44}, we need to add a correction parameter $\bar c$ in \eqref{eq60}. Then the necessary condition of $N$ for the full-DoF can be improved as
\begin{equation}
\begin{array}{l}
N < \left( {\sqrt L  + c} \right)\eta  + \bar c
\end{array}
\label{eq61}.  
\end{equation}

\section{Sum Rate Maximization  }
Since the DoF serves as a first-order approximation of system capacity, this section proposes a sum-rate maximization problem for the RIS-aided multi-cell system to validate the conclusions drawn in the previous sections.

According to \eqref{eq22}, 
the rate of each user can be expressed as \eqref{eq75} in Appendix \ref{appF}.
Therefore, the sum-rate maximization problem can be denoted as
\begin{equation}
\begin{aligned}
\mathop {\max }\limits_{\bf{v}} \; &{\rm{ W = }}\sum\limits_{i = 1}^G {\sum\limits_{k = 1}^K { R_k^{\left[ i \right]}} } \\
{\rm{s}}.{\rm{t}}.&\left| {{v_n}} \right| = 1,n \in \left[ {1,N} \right]
\end{aligned}
\label{eq76}.    
\end{equation}

Since ${\bf{v}}$ lies in a complex circle manifold $\mathcal{M} = \left\{ {\left. {{\mathbf{v}} \in {\mathbb{C}^{N \times 1}}} \right|\left| {{v_1}} \right| =  \cdots  = \left| {{v_N}} \right| = 1} \right\}$, we can use manifold optimization \cite{boumal2023introduction} such as RCG to solve ${\bf{v}}$.
In each iteration, the RCG algorithm contains three key steps, namely, Riemannian gradient calculation, transport and retraction \cite{li2022joint,jiang2022interference}.

The Riemannian gradient ${\text{grad}}W\left( {\mathbf{v}} \right)$ is the orthogonal projection of the Euclidean gradient \cite{hjorungnes2007complex},\cite{petersen2008matrix} $ \nabla W\left( {\mathbf{v}} \right)$ onto the tangent space of the manifold, i.e.,
\begin{equation}
\begin{array}{l}
{\text{grad}}W\left( {\mathbf{v}} \right){\text{ = }}\nabla W\left( {\mathbf{v}} \right) - {\text{re}}\left( {\nabla W\left( {\mathbf{v}} \right) \circ {{\mathbf{v}}^*}} \right) \circ {\mathbf{v}}
\end{array}
\label{eq70},    
\end{equation}
where $ \circ $ denotes the element-wise product, 
and $\nabla W\left( {\mathbf{v}} \right)$ is expressed as \eqref{eq77} in the Appendix \ref{appF}.

Then the search direction d is updated by
\begin{equation}
{{\tilde d}^t} = {\text{grad}}W\left( {{{\mathbf{v}}^t}} \right) + {\lambda _1}{{\rm T}_{{{\mathbf{v}}^t}}}\left( {{{\tilde d}^{t - 1}}} \right)
\label{eq72},       
\end{equation}
where ${\lambda _1}$ is the the Polak-Ribiere parameter and the transport operation is
\begin{equation}
{{\rm T}_{{{\mathbf{v}}^t}}}\left( {{{\tilde d}^{t - 1}}} \right) = {{\tilde d}^{t - 1}} - {\text{re}}\left( {{{\tilde d}^{t - 1}} \circ {{\left( {{{\mathbf{v}}^t}} \right)}^*}} \right) \circ {{\mathbf{v}}^t}
\label{eq73}.     
\end{equation}

In order to keep the updated ${\mathbf{v}}$ on the manifold $\mathcal{M}$, we should perform a retraction on it, i.e., element-wise normalization
\begin{equation}
{{\mathbf{v}}^{t + 1}} = {{\left( {{{\mathbf{v}}^t} + {\lambda _2}{{\tilde d}^t}} \right).} \mathord{\left/
		{\vphantom {{\left( {{{\mathbf{v}}^t} + {\lambda _2}{{\tilde d}^t}} \right).} {abs\left( {{{\mathbf{v}}^t} + {\lambda _2}{{\tilde d}^t}} \right)}}} \right.
		\kern-\nulldelimiterspace} {abs\left( {{{\mathbf{v}}^t} + {\lambda _2}{{\tilde d}^t}} \right)}}
\label{eq74},	
\end{equation}
where ${{\lambda _2}}$ is the Armijo step size. The overall algorithm is summarized in Algorithm \ref{algorithm1}.
\begin{algorithm}
	\caption{RCG algorithm for solving problem \eqref{eq76}}
	\begin{algorithmic}
		\STATE 
		\STATE \textbf{Initialization:} 
  set the result of the alternating projection algorithm \cite{jiang2022interference} for interference null equations \eqref{eq25} as the initial ${{\mathbf{v}}^0}$, and $\varepsilon  = {10^{ - 3}}$
		
		\STATE \textbf{Iteration:}

		\STATE \hspace{0.5cm}
		1 calculate Riemannian gradient by \eqref{eq70}
		\STATE \hspace{0.5cm}
		2 update the search direction by \eqref{eq72}
		\STATE \hspace{0.5cm}
		3 retract ${\mathbf{v}}$ onto the manifold by \eqref{eq74}

		\STATE \textbf{Stop:} until $\frac{{\left| {{W^{(t + 1)}} - {W^{(t)}}} \right|}}{{{W^{(t)}}}} \le \varepsilon $ or the maximum number of
		iterations is reached
		
		\STATE \textbf{Output:} $W$ and ${\bf{v}}$
	\end{algorithmic}
	\label{algorithm1}
\end{algorithm}
After getting the optimal ${\bf{v}}$, we can use the following formulation to calculate the DoF of each cell.

\begin{equation}
\begin{array}{l}
{\rm{Do}}{{\rm{F}}_i} = \sum\limits_{k = 1}^K {\mathop {\lim }\limits_{SNR_k^{[i]} \to \infty } \frac{{R_k^{[i]}}}{{\log \left( {{\rm{SNR}}_k^{[i]}} \right)}}} ,
\end{array}         
\end{equation}
where ${{\rm{SNR}}_k^{[i]}}$ is the received SNR of the $k$-th user in cell $i$.

\section{Simulation Result}

\begin{figure*}[b]
	\captionsetup[subfloat]{font=scriptsize}
	\centering\subfloat[The total rate vs. the number of RIS elements]{\includegraphics[width=7.5cm]{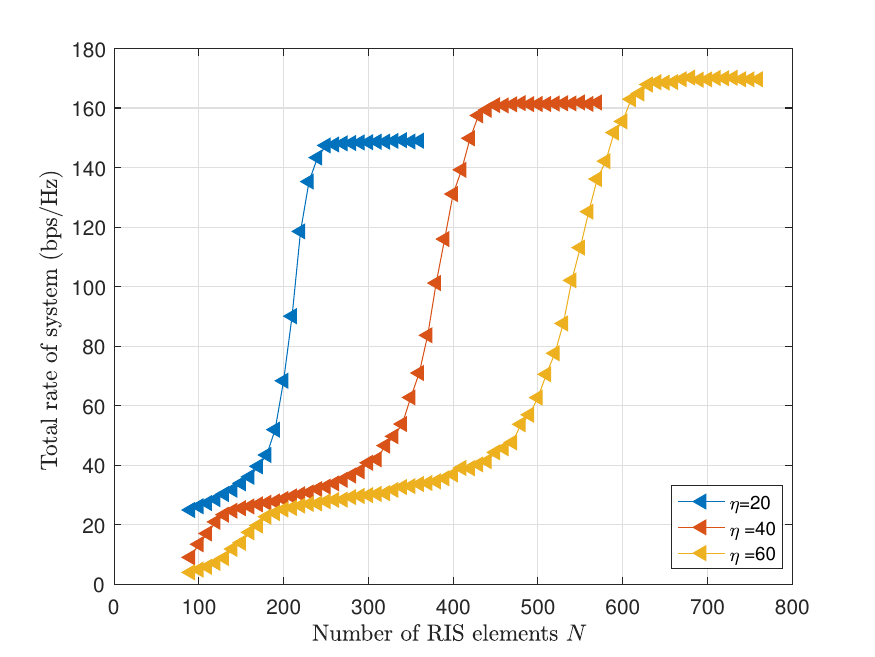}%
		\label{RATE_4}}
	\hfil
	\centering\subfloat[The total DoF vs. the number of RIS elements]{\includegraphics[width=7.5cm]{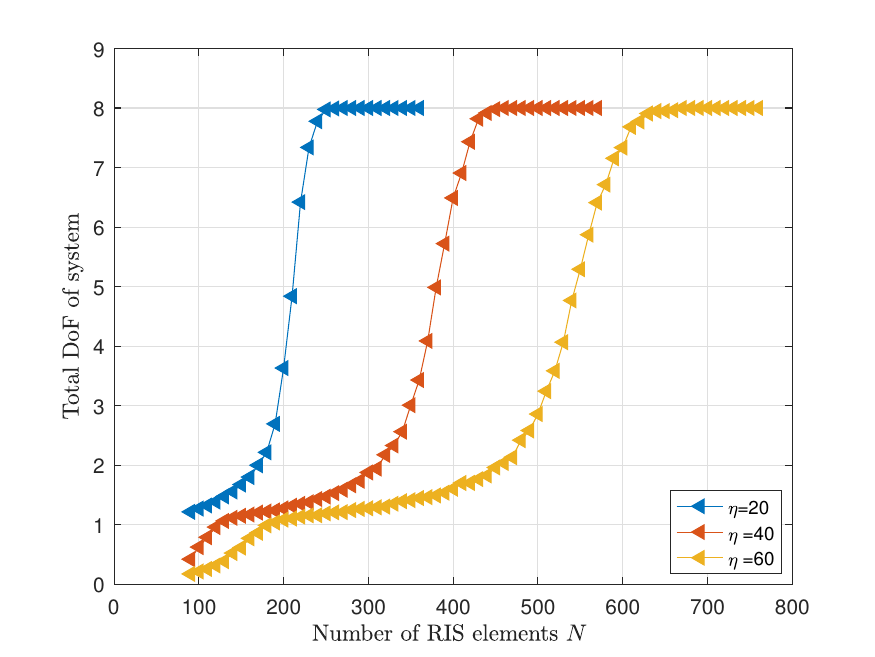}%
		\label{DoF_4}}
	\caption{ The total rate and DoF vs. the number of RIS elements $N$ under the different strength ratio $\eta $ between direct and reflective links ($G=2$, $M=K=4$). }
	\label{fig3}
\end{figure*}

In this section, we present numerical results to quantify the required $N$ for the full-DoF, and identify the factors affecting $N$ and their impact. Since the intensity disparity between direct links ${\mathbf{b}}$ and cascaded reflective links ${\mathbf{A}}$ is the main factor affecting $N$, we first consider the case where all users have the same path loss. The cases with different path losses among users are discussed at the end of this section.

Without loss of generality, we assume that there are two cells, i.e., $G=2$, and the transmission power of each user is the same, i.e., $p_k^{[g]} = P,\forall g \in \left[ {1,2} \right],{\rm{ }}k \in \left[ {1,K} \right]$. 
The system bandwidth is $1{\rm{MHz}}$ and the noise spectral density is $ - {\rm{174dBm/Hz}}$.
We further assume that all channels follow Rayleigh fading, and the RIS is modeled as a linear array for convenience. Specifically, the reflect channel ${\bf{h}}_{Ik}^{[i]} \sim {\cal C}{\cal N}\left( {0,\sigma _1^2{{\bf{I}}_N}} \right)$, each element of ${{\bf{H}}_i}$ is i.i.d.$ \sim {\cal C}{\cal N}\left( {0,\sigma _2^2} \right)$, and the direct channel ${\mathbf{h}}_{{B_i}k}^{[j]} \sim {\cal C}{\cal N}\left( {0,\sigma _3^2{{\bf{I}}_M}} \right)$. And $\sigma _1^2 = \sigma _2^2 =  - 30{\rm{dBm}}$, ${\sigma _3} = \eta {\sigma _1}{\sigma _2}$, $P = 1w$.

\subsection{Sum Rate and DoF Maximization}
The relationship between sum-rate (DoF), strength ratio $\eta$ of the direct link to the cascaded reflective link and the number of RIS elements $N$ is illustrated in Fig. \ref{fig3},
where $G=2$, $M=K=4$. From Fig. \ref{fig3},  we can see that, 
when $N$ is very small, the total rate increases very slowly, and the DoF is much smaller than $G \cdot \min \left\{ {M,K} \right\} = 8$. This is because at this stage RIS cannot manage interference and can only play a role in building reflective links.  
Then, as $N$ further increases, the total rate increases very quickly. This is because at this stage, RIS can not only be used to reflect signals but also to eliminate interference. Finally, when $N$ exceeds a certain threshold, the total rate almost no longer increases. This is because the system has obtained the full-DoF and entered the noise-limited state. At this stage, we can only increase the system capacity by improving the SNR, such as increasing the transmission power. From Fig.\ref{DoF_4} we can also see that the required $N$ for the full-DoF increases with the strength of the direct link. This is because when the direct wave is strong, more reflective elements are needed to amplify the reflected wave and make it comparable to the direct wave.

\subsection{The Oder of $N$}
\label{geometry1}
To verify the analysis of Section \ref{Sufficient_B}, in Fig. \ref{fig5_1}, we plot the situation of ${S_1} \cap {S_2}$ and ${S_1} \cap {{\hat S}_2}$. And 200 independent trials are conducted. The '${\rm{geometry}}$' line denotes the value of ${N_1}$ in \eqref{eq32_1}. The '${\rm{gordon}}\_{\rm{cir}}$' line means the necessary condition for ${S_1} \cap {{\hat S}_2} \ne \phi $ (i.e., the ${N_2}$ of \eqref{eq34_1}), while the '${\rm{gordon}}\_{\rm{modulu}}\_1$' line denotes the necessary condition \eqref{eq39} for ${S_1} \cap {S_2} \ne \phi $. The line '$99\%$ solved' in Fig. \ref{int_cir}  indicates that if $N$ exceeds the corresponding value, ${\bf{v}} \in {S_1} \cap {{\hat S}_2} \ne \phi $ can be guaranteed with a probability of $99\%$.	
From Fig. \ref{int_cir}, we can see that the condition for $50\%$ probability of ${S_1} \cap {{\hat S}_2} \ne \phi $ is symmetrically positioned between the $99\%$ line and the $1\%$ line. 
From Fig. \ref{int_modulus1}, we can see that the sufficient condition \eqref{eq35_1} for ${S_1} \cap {{\hat S}_2} \ne \phi $ can also be approximated as the sufficient condition for ${S_1} \cap {S_2} \ne \phi $. Moreover, the '${\rm{geometry}}$' line is very close to the $1\%$ line of ${S_1} \cap {S_2} \ne \phi $. This is because when the sphere ${{\hat S}_2}$ intersects exactly with ${S_1}$, it marks the critical point at which ${S_2}$ and ${S_1}$ may intersect. In other words, if the radius of the sphere is smaller than the corresponding value at this critical point, 
${S_1} \cap {S_2} \ne \phi $ is almost imposible.

\begin{figure*}
	\captionsetup[subfloat]{font=scriptsize}
	\centering\subfloat[Conditions for the intersection of ${S_1}$ and ${{\hat S}_2}$ ]{\includegraphics[width=7.5cm]{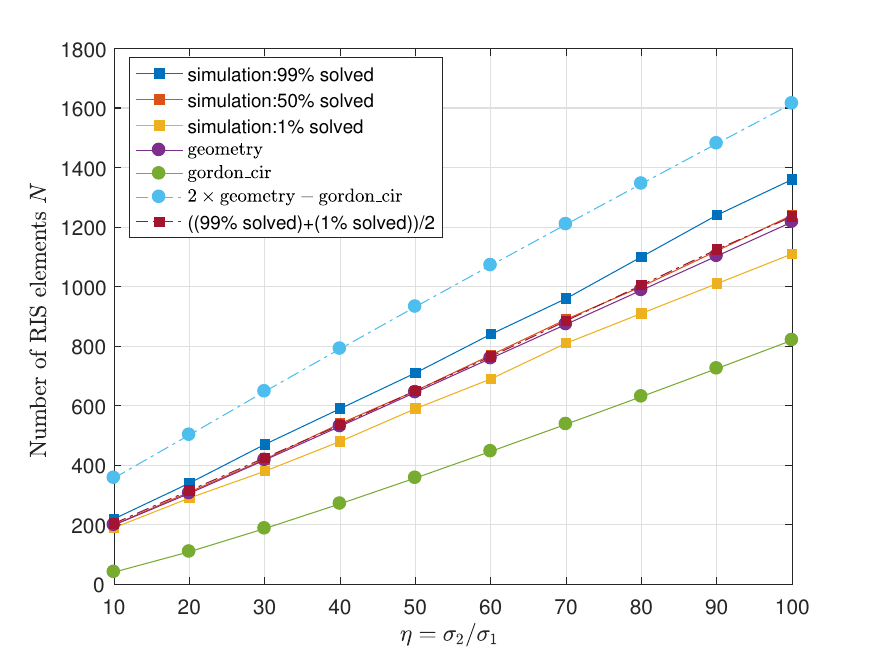}%
		\label{int_cir}}
	\hfil
	\centering\subfloat[Conditions for the intersection of ${S_1}$ and ${S_2}$]{\includegraphics[width=7.5cm]{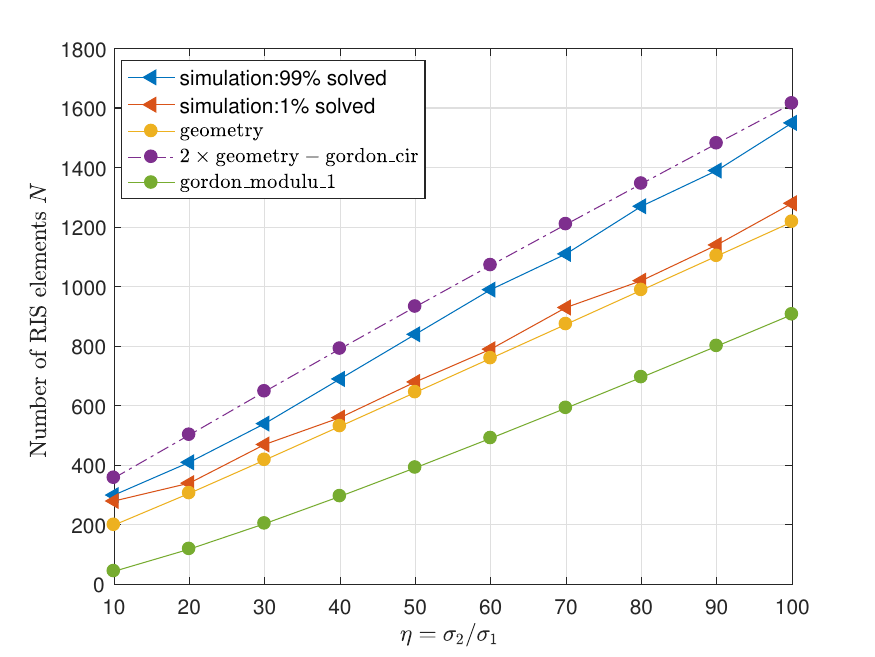}%
		\label{int_modulus1}}
	
	\caption{ The order of sufficient and necessary conditions, ($G=2$, $M=K=6$ ). }
	
	\label{fig5_1}
\end{figure*}

\begin{figure*}[t]
	\captionsetup[subfloat]{font=scriptsize}
	\centering\subfloat[Full-DoF with probability of $1{\rm{\% }}$ when $\eta $ is large]{\includegraphics[width=7.5cm]{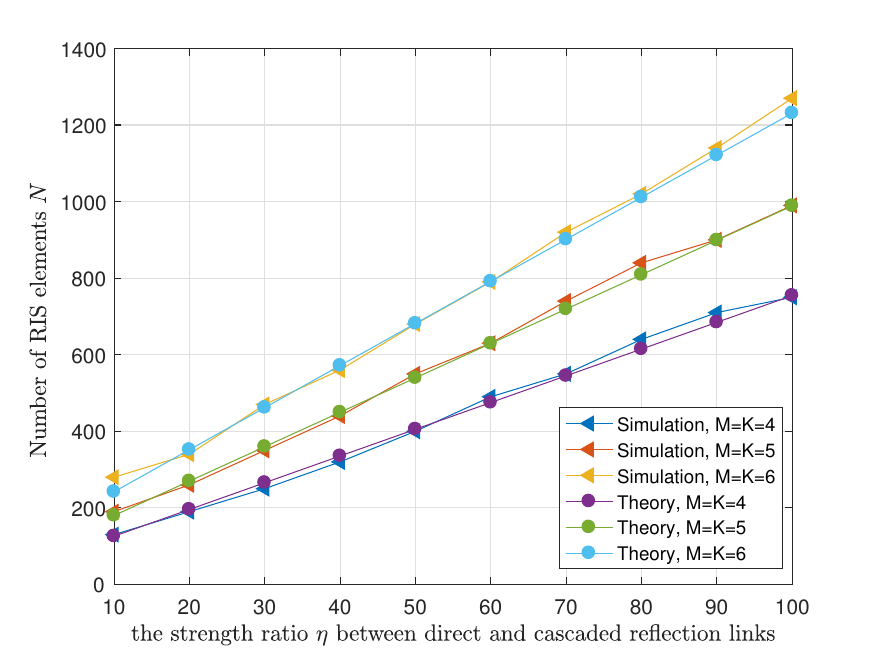}%
		\label{necessary}}
	\hfil
	\centering\subfloat[Full-DoF with probability of $1{\rm{\% }}$ when $\eta $ is small]{\includegraphics[width=7.5cm]{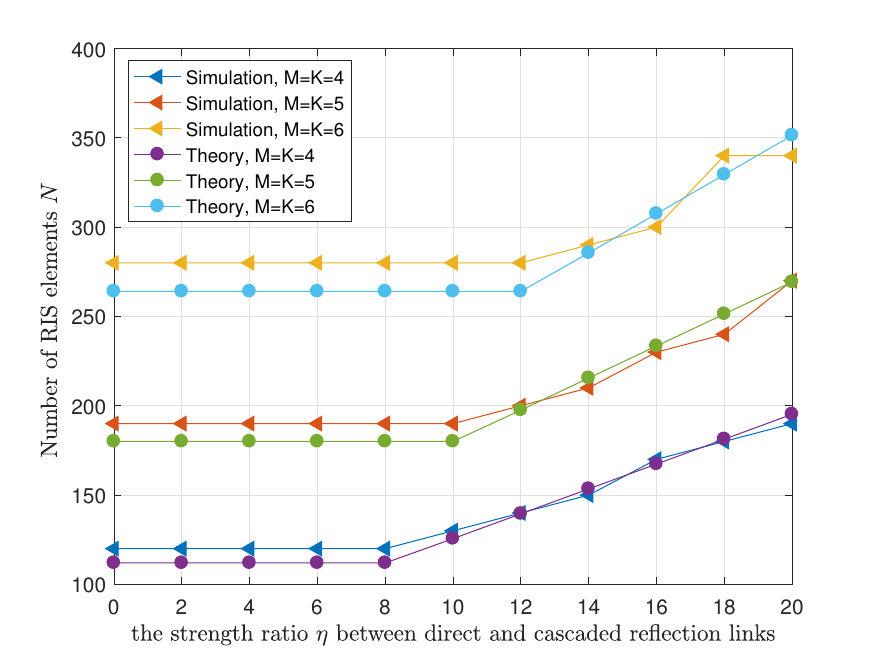}%
		\label{necessary_smal}}
	
	\caption{ The necessary conditions of the required RIS elements $N$ for full-DoF, ($G=2$). }
	
	\label{fig5}
\end{figure*}

\subsection{More Precise Necessary $N$}
To quantify the necessary $N$ for the full-DoF as derived in section \ref{Precise}, in Fig. \ref{fig5}, we consider the case of $M=K=4,5,6$ with $\eta $ ranging from 10 to 100.
'Full-DoF with probability of $1{\rm{\% }}$' means that if $N$ is less than the corresponding threshold, it is almost impossible to obtain the full-DoF.  
From \eqref{eq42} we can see that $\eta$ is the dominant factor affecting $N$ when $\eta$ is large enough, otherwise, $L$ is dominant.  Therefore, we speculate that the correction parameters ${\bar c}$ in \eqref{eq61} is approximately $L$, and the corresponding threshold can be denoted as
\begin{equation}
\begin{array}{l}
\left( {\sqrt L  + c} \right)\eta+L 
\end{array}
\label{eq89},    
\end{equation}
where $L{\text{ = }}GM\left( {GM - 1} \right)$. When $c \approx  - 0.5$, this value of \eqref{eq89}, plotted as the 'Theory' line in Fig. \ref{necessary}, can fit the actual value very well.

Moreover, since the interference cancellation equations \eqref{eq26}, consisting of $L$ complex equations and $N$ real variables (i.e., $N$ phase shift $\theta $ of the passive RIS), require at least $2L$ variables to ensure solvability. Thus, the minimal $N$ can not less than $2L$, then the threshold \eqref{eq89} can be improved as
\begin{equation}
\begin{array}{l}
\max \left\{ {2L, \left( {\sqrt L  + c} \right)\eta +L } \right\}
\end{array}
\label{eq90}.   
\end{equation}
The turning point of $\eta$ can be calculated as
\begin{equation}
\begin{array}{l}
\eta  \approx \frac{L}{{\sqrt L  + c}} \approx \sqrt L  
\end{array}
\label{eq91}.    
\end{equation}
Therefore, the $2k\left( {k - 1} \right)$ criterion of \cite{jiang2022interference} can be regarded as a subset of our result. This is because the $k$-user interference channel is equivalent to a cellular network with $G=k$ cells, each containing $K=1$ single-antenna user, and a BS equipped with $M=1$ antenna.

Fig. \ref{necessary_smal} illustrates the case of smaller $\eta $, where the 'Theory' line is given by \eqref{eq90}.
We can see that when $\eta $ is very small, the required $N$ is the same as when there is no direct link (i.e., $\eta=0 $ ). However, once $\eta $ exceeds a certain value, the required $N$ increases linearly with $\eta$. And this phase transition point of $\eta$ is approximately ${\sqrt L }$.

\subsection{Users With Different Path Losses}
\label{different_path}
Considering that in practice, the location of users is randomly changing. In Fig. \ref{fig7}, we assume that the RIS is positioned at $\left( {{\rm{0,0,0}}} \right)$ meters, BS1 at $\left( {{\rm{15,15,0}}} \right)$, and BS2 at $\left( {{\rm{40,15,0}}} \right)$. The users in cell-1 and cell-2 are uniformly distributed within the rectangular regions $\left[ {5,25} \right] \times \left[ { - 25, - 5} \right]$ and $\left[ {30,50} \right] \times \left[ { - 25, - 5} \right]$, respectively, with the same z-coordinates ($z=-20$). The path loss model is $los\left( d \right) = {T_0} \times {\left( d \right)^{ - \alpha }}$, where $d$ is the link distance in meters, ${T_0} =  - 30{\text{dB}}$ is the path loss at a reference distance, and $\alpha $ is the path loss exponent. For reflect links (user-RIS, RIS-BS), $\alpha $ is set to 2, while for the direct link (user-BS), $\alpha $ is set to 4\cite{oa2018determination}. 
\begin{figure}
	\centering\includegraphics[width=7.5cm,clip ]{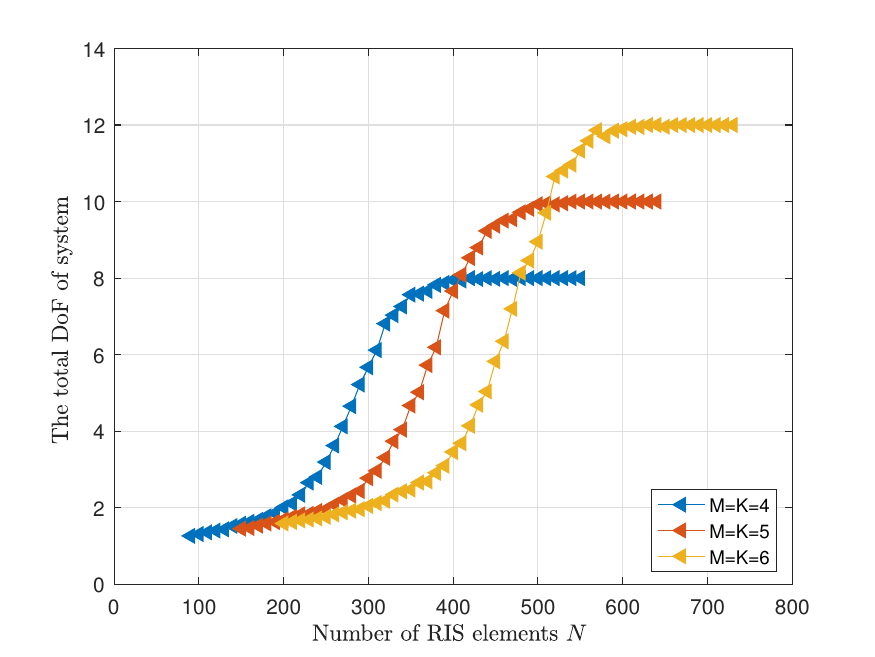}
	\caption{The DoF vs. the number of RIS elements when the user's location randomly changes ($G=2$).}
	\label{fig7}
\end{figure}
From Fig. \ref{fig7}, we can also see that when $N$ is greater than a certain value, the full-DoF can be obtained. Since the average positions of users in cell-1 and cell-2 are $\left[ {15, - 15, - 20} \right]$ and $\left[ {40, - 15, - 20} \right]$ respectively, the average $\eta $ values are 15 and 49 respectively. Thus, the overall average  $\eta $ is 32. Substituting $\eta=32 $ into \eqref{eq35_1}, we obtain the theoretical sufficient $N$ as $391$, $528$ and $677$ for $M=K=4,5,6$, respectively. Similarly, substituting it into \eqref{eq39}, we find the theoretical necessary $N$ as $163$, $194$ and $222$ for $M=K=4,5,6$, respectively.
As shown in Fig. \ref{fig7}, the actual values of $N$ fall within the range of our theoretical findings.

\section{Conclusion}
This paper investigates an RIS-aided multi-cell MIMO system. We show that, in the practical $G$-cellular system with more single-antenna users $K$ than BS antennas $M$ per cell, the full-DoF $\min \left\{ {M,K} \right\}=M$ can be achieved by each cell if the number of RIS elements $N$ is sufficiently large. 
To identify the factors affecting $N$ and their impact, we first assume uniform path losses for the direct and cascaded reflective links respectively.
Subsequently, we use Gordon's Theorem and the geometry relationship between the interference null equations and the modulus-1 constraints to derive the necessary and sufficient conditions, respectively. It has been proved that both of them have the same form of threshold. Specifically, if the strength ratio $\eta$ of the direct link to the cascaded reflective link is large enough, the required RIS elements $N$ for achieving the full-DoF is not only linearly related to $\eta$ but also dependent on ${\sqrt L }$ (where $L{\rm{  =  }}GM\left( {GM - 1} \right)$). Conversely, if $\eta$ is smaller, $N$ depends only on $L$. 
Additionally, a more precise necessary threshold for $N$, given by $\max \left\{ {2L, \left( {\sqrt L  + c} \right)\eta +L } \right\}$ (where $c$ is a constant), is derived using the extreme value statistics. We show that the phase transition point of $\eta$, which causes a significant increase in $N$, is approximately ${\sqrt L }$.

We also propose a sum rate maximization problem and adopt the RCG algorithm to solve it. Simulation results show that our theoretical findings are highly consistent with actual performance and are applicable to the practical case with different path losses for each user. These findings can provide valuable theoretical guidance for the practical system design.

\appendix
\subsection{Proof of Theorem \ref{the1}}	
\label{appA}
If the BS can apply a decoding matrix to the received signals by antenna collaboration, and each user's message is decoded from a different row of ${\bf{U}}_i^H{{\bf{y}}_i}$. Then the message of user $k$ in cell $i$ can be decoded as
\begin{equation}
\begin{aligned}
{{{\mathbf{\bar y}}}_{ik}}{\text{ = }}&{\mathbf{u}}_{ik}^H\left( {{{\left( {{\mathbf{A}}_{ik}^{[i]}} \right)}^H}{\mathbf{v}} + {\mathbf{h}}_{{B_i},k}^{[i]}} \right)x_k^{[i]} \hfill \\
&+ \underbrace {{\mathbf{u}}_{ik}^H\sum\limits_{j \ne k}^K {\left( {{{\left( {{\mathbf{A}}_{ij}^{[i]}} \right)}^H}{\mathbf{v}} + {\mathbf{h}}_{{B_i},j}^{[i]}} \right)x_j^{[i]}} }_{{\text{intra - cell interference}}} \hfill \\
&+ \underbrace {{\mathbf{u}}_{ik}^H\sum\limits_{g \ne i}^G {\sum\limits_{j = 1}^K {\left( {{{\left( {{\mathbf{A}}_{ij}^{[g]}} \right)}^H}{\mathbf{v}} + {\mathbf{h}}_{{B_i},j}^{[g]}} \right)x_j^{[g]}} } }_{{\text{inter - cell interference}}} + {{\mathbf{z}}_{ik}} \hfill \\ 
\end{aligned}
\label{eq6}, 
\end{equation}
where ${{\bf{U}}_i} = \left[ {{{\bf{u}}_{i1}}, \cdots ,{{\bf{u}}_{iK}}} \right]$ is the decoding matrix adopted by BS of cell $i$  and ${{\mathbf{u}}_{ik}} \ne {\mathbf{0}} \in {\mathbb{C}^{M \times 1}}$, 
${{\bf{z}}_{ik}}$ is the $k$-th row of ${\bf{U}}_i^H{{\bf{n}}_i}$.
To achieve the full-DoF of each cell, all the intra-cell and inter-cell interference need to be eliminated. Specifically, we need to solve the following problem 
\begin{subequations}
\begin{align}
{\rm{find }}&{\bf{v}},{{\bf{u}}_{ik}},i \in \left[ {1,G} \right],{\rm{ }}k \in \left[ {1,K} \right] \notag \\
{\rm{s}}{\rm{.t}}{\rm{.   }}&{\mathbf{u}}_{ik}^H\left( {{{\left( {{\mathbf{A}}_{ik}^{[i]}} \right)}^H}{\mathbf{v}} + {\mathbf{h}}_{{B_i},k}^{[i]}} \right)x_k^{[i]} \ne 0 \label{10a} \\
&{\bf{u}}_{ik}^H\left( {{{\left( {{\bf{A}}_{ij}^{[i]}} \right)}^H}{\bf{v}} + {\mathbf{h}}_{{B_i},j}^{[i]}} \right)x_j^{[i]} = 0,{\rm{  }}\forall j \ne k \label{10b} \\
&{\bf{u}}_{ik}^H\left( {{{\left( {{\bf{A}}_{ij}^{[g]}} \right)}^H}{\bf{v}} + {\mathbf{h}}_{{B_i},j}^{[g]}} \right)x_j^{[g]} = 0,{\rm{ }}j \in \left[ {1,K} \right]{\rm{,}}\forall {\rm{ g}} \ne i  \label{10c}\\
&\left| {{v_n}} \right| = 1,n \in \left[ {1,N} \right]  \label{10d}
\end{align}	
\label{eq10} 
\end{subequations}
where \eqref{10d} represents the modulus-1 constraint imposed by the passive RIS.

For each cell $i$, in order to eliminate the intra-cell interference while preserving the desired messages, \eqref{10a} and \eqref{10b} should hold simultaneously, which can be reformulated as
\begin{equation}
{\left[ {\begin{array}{*{20}{c}}
		{{\mathbf{u}}_{i1}^H} \\ 
		{{\mathbf{u}}_{i2}^H} \\ 
		\vdots  \\ 
		{{\mathbf{u}}_{iK}^H} 
		\end{array}} \right]_{K \times M}} \times {\left[ {\begin{array}{*{20}{c}}
		{{{{\mathbf{\tilde v}}}_{i1}}}& \cdots &{{{{\mathbf{\tilde v}}}_{iK}}} 
		\end{array}} \right]_{M \times K}} = dia{g_{K \times K}}
\label{eq11},	
\end{equation}
where ${{{\mathbf{\tilde v}}}_{ik}} = \left( {{{\left( {{\mathbf{A}}_{ik}^{[i]}} \right)}^H}{\mathbf{v}} + {\mathbf{h}}_{{B_i},k}^{[i]}} \right)x_k^{[i]} \in {\mathbb{C}^{M \times 1}}$, and $dia{g_{K \times K}}$ is a $K$-diagonal matrix in which all elements outside the main diagonal are zero.
If we denote the equation \eqref{eq11} as 
\begin{equation}
\begin{array}{l}
{\mathbf{U}} \times {\mathbf{\tilde V}} = {\mathbf{D}},
\end{array}
\label{eq12}    
\end{equation}
then we have
\begin{equation}
K = rank(D) \le \min \left\{ {rank({\bf{U}}),rank({\mathbf{\tilde V}})} \right\}
\label{eq13}.   
\end{equation}
Due to the randomness of the direct channel ${\mathbf{h}}_{{B_i},k}^{[i]}$, even without the term of ${\left( {{\mathbf{A}}_{ik}^{[i]}} \right)^H}{\mathbf{v}}$ (i.e., without RIS), we can still have
\begin{equation} 
rank({\mathbf{\tilde V}}) = \min \left\{ {M,K} \right\}.
\label{eq14}    
\end{equation} 
Combining \eqref{eq13}, \eqref{eq14} and the basic properties of matrix rank 
\begin{equation}
rank({\bf{U}}) \le \min \left\{ {M,K} \right\},
\label{eq15}    
\end{equation}
we can get 
\begin{equation}
K = rank(D) \le rank({\bf{U}}) \le \min \left\{ {M,K} \right\}
\label{eq16}.   
\end{equation}
Thus
\begin{equation}
M \ge K
\label{eq17}.   
\end{equation}

For the full DoF, the inter-cell interference \eqref{10c} also needs to be eliminated, which can be reformulated as 
\begin{equation} 
{\left[ {\begin{array}{*{20}{c}}
		{{\mathbf{u}}_{i1}^H} \\ 
		{{\mathbf{u}}_{i2}^H} \\ 
		\vdots  \\ 
		{{\mathbf{u}}_{iK}^H} 
		\end{array}} \right]_{K \times M}} \times {\left[ {\underbrace {\begin{array}{*{20}{c}}
			{{{{\mathbf{\tilde v}}}_{g1}}}& \cdots &{{{{\mathbf{\tilde v}}}_{gK}}} 
			\end{array}}_{g \ne i,g \in \left[ {1,G} \right]}} \right]_{M \times K\left( {G - 1} \right)}} = {\mathbf{0}}
\label{eq18}    
\end{equation}
where ${{{\mathbf{\tilde v}}}_{gk}} = \left( {{{\left( {{\mathbf{A}}_{ik}^{[g]}} \right)}^H}{\mathbf{v}} + {\mathbf{h}}_{{B_i},k}^{[g]}} \right)x_k^{[g]} \in {\mathbb{C}^{M \times 1}}$. For ease of expression, we also denote the above equation \eqref{eq18} as 
\begin{equation} 
{\bf{U}} \times {\bf{\bar V}} = {\bf{0}}.
\label{eq19}    
\end{equation} 
Since the channel realizations ${\mathbf{h}}_{{B_i},k}^{[j]}$ are random, even without the term of ${\left( {{\mathbf{A}}_{ik}^{[g]}} \right)^H}{\mathbf{v}}$ (i.e., without RIS), we still have $rank\left( {{\mathbf{\bar V}}} \right) = \min \left\{ {M,K\left( {G - 1} \right)} \right\}$. Therefore, the equation \eqref{eq18} implies that ${\bf{U}}$ has $\min \left\{ {M,K\left( {G - 1} \right)} \right\}$ linearly independent non-zero solutions, which means that the null space of ${\bf{U}}$ is $\min \left\{ {M,K\left( {G - 1} \right)} \right\}$-dimensional, i.e.,
\begin{equation}
M - rank({\mathbf{U}}) = \min \left\{ {M,K\left( {G - 1} \right)} \right\}
\label{eq20}   
\end{equation}
According to \eqref{eq16} and \eqref{eq20}, we have 
\begin{equation}
M = K\left( {G - 1} \right) + rank({\bf{U}}) \ge KG
\label{eq21}.   
\end{equation}
Furthermore, if the direct channel is blocked, i.e., ${{{\mathbf{\tilde v}}}_{jk}} = {\left( {{\mathbf{A}}_{ik}^{[j]}} \right)^H}{\mathbf{v}}x_k^{[j]}$, \eqref{eq16} and \eqref{eq20} are still hold due to the randomness of the cascaded channel ${{\mathbf{A}}_{ik}^{[j]}}$. Therefore, in the case of no direct channel and $M \ge KG$, RIS is only used to create the communication link between users and the BS, and a small amount of reflective elements is sufficient to achieve the full DoF.

The above analyses indicate that the necessary condition for achieving full DoF solely through antenna collaboration is $M \ge KG$. Conversely, if $M < KG$, the full DoF cannot be achieved by antenna collaboration alone. Therefore, in the case of $M < KG$, to ensure the equations \eqref{eq12} and \eqref{eq19} hold simultaneously, we must change the structure of ${{\bf{\tilde V}}}$ and ${{\bf{\bar V}}}$ (such as reducing the rank of ${{\bf{\tilde V}}}$ and ${{\bf{\bar V}}}$ by adjusting ${\bf{v}}$). In other words, if $M < KG$,  the assistance of RIS is essential.

\subsection{Proof of Theorem \ref{the2}}	
\label{appB}
Because ${{\mathbf{Gx}}}$ has the same distribution with $\sigma \sqrt N {\mathbf{\tilde G\tilde x}}$, i.e., 
\begin{equation}
\begin{array}{l}
{\mathbf{Gx}} \cong \sigma \sqrt N {\mathbf{\tilde G\tilde x}}
\end{array}
\label{eq30},    
\end{equation}
where each element of ${\mathbf{\tilde G}} \in {\mathbb{C}^{N \times L}}$ is i.i.d.$ \sim \mathcal{C}\mathcal{N}\left( {0,1} \right)$ , and ${\mathbf{\tilde x}} \in \tilde S$ is a closed subset of the unit sphere in $N$ dimensions,
\begin{equation}
\begin{array}{l}
\tilde S = \left\{ {\left. {{\mathbf{\tilde x}}} \right|\left| {{{\tilde x}_i}} \right| = \frac{1}{{\sqrt N }},\forall i = 1, \cdots ,N} \right\} \in {\mathbb{S}^{N - 1}}
\end{array}
\label{eq32}.    
\end{equation}

Thus we have 
\begin{equation}
E\left( {\left| {{\mathbf{Gx}}} \right|} \right) = \sigma \sqrt N  \cdot E\left( {\left| {{\mathbf{\tilde G\tilde x}}} \right|} \right)
\label{eq31}.   
\end{equation}

According to the Gordon's Theorem \cite{bandeira2015ten}, we can get
\begin{equation}
\begin{array}{l}
E\left( {\min \left| {{\bf{\tilde G\tilde x}}} \right|} \right) \ge {a_L} - W\left( {\tilde S} \right)\\
E\left( {\max \left| {{\bf{\tilde G\tilde x}}} \right|} \right) \le {a_L} + W\left( {\tilde S} \right)
\end{array}
\label{eq33},    
\end{equation}
where ${a_L} = \sqrt L $ and $W\left( {\tilde S} \right)$ is the Gaussian width \cite{vershynin2018high} of the set ${\tilde S}$, it can be calculated as
\begin{equation}
\begin{aligned}
W\left( {\tilde S} \right) &= E\left( {\mathop {\sup }\limits_{{\mathbf{v}} \in \tilde S} \left| {\left\langle {{\mathbf{g}},{\mathbf{\tilde x}}} \right\rangle } \right|} \right) \hfill \\
&= E\left( {\left\{ {\begin{aligned}
		&{\max {\rm{ }}\left| {{{\bf{g}}^H}{\bf{\tilde x}}} \right| = \left| {\sum\limits_{i = 1}^N {g_i^*{{\tilde x}_i}} } \right|}\\
		&{s.t.{\rm{ }}\left| {{{\tilde x}_i}} \right| = \frac{1}{{\sqrt N }}{\rm{  }}}
		\end{aligned}} \right.} \right) \\
&= E\left( {\frac{1}{{\sqrt N }}\sum\limits_{i = 1}^N {\left| {{g_i}} \right|} } \right) = \sqrt N \frac{{\sqrt \pi  }}{2} \hfill \\ 
\end{aligned}
\label{eq34}    
\end{equation}
where  ${\mathbf{g}} \sim \mathcal{C}\mathcal{N}\left( {0,{{\mathbf{I}}_N}} \right)$ is a gaussian vector.
Combining \eqref{eq31} and \eqref{eq33} we can get \eqref{eq29},
which completes the proof.

\subsection{Proof of Theorem \ref{the4}}	
\label{appD}
\begin{equation}	
\begin{aligned}
E\left( {\left| {{{\bf{G}}^ + }{\bf{x}}} \right|} \right) &= E\left( {\sqrt {{{\left( {{{\bf{G}}^ + }{\bf{x}}} \right)}^H}{{\bf{G}}^ + }{\bf{x}}} } \right)\\
&\le \sqrt {E\left( {{{\bf{x}}^H}{{\left( {{{\bf{G}}^ + }} \right)}^H}{{\bf{G}}^ + }{\bf{x}}} \right)} \\
&= \sqrt {E\left( {tr\left( {{\bf{x}}{{\bf{x}}^H}{{\left( {{\bf{G}}{{\bf{G}}^H}} \right)}^{ - 1}}} \right)} \right)} \\
&= \sqrt {\sigma _2^2tr\left( {E\left( {{{\left( {{\mathbf{G}}{{\mathbf{G}}^H}} \right)}^{ - 1}}} \right)} \right)}  \\
&= \sqrt {\frac{L}{{N - L - 1}}} \rho 
\end{aligned}
\label{add3},    
\end{equation}
where ${{\bf{G}}^ + } = {{\bf{G}}^H}{\left( {{\bf{G}}{{\bf{G}}^H}} \right)^{ - 1}}$, and ${\left( {{\bf{G}}{{\bf{G}}^H}} \right)^{ - 1}}$ follows an Inverse-Wishart distribution, which completes the proof.

\subsection{Proof of Theorem \ref{the5}}	
\label{appE}
Just like the proof of Theorem \ref{the2}, we have 
\begin{equation}
E\left( {\left| {{\bf{Gx}}} \right|} \right) = \sigma \sqrt N  \cdot E\left( {\left| {{\bf{\tilde G\hat x}}} \right|} \right)
\label{eq96},    
\end{equation}
where each element of ${\mathbf{\tilde G}} \in {\mathbb{C}^{N \times L}}$ is i.i.d.$ \sim \mathcal{C}\mathcal{N}\left( {0,1} \right)$ , and ${\bf{\hat x}} \in \hat S$ is on the unit sphere in $N$ dimensions, i.e.,
\begin{equation}
\hat S = \left\{ {\left. {{\mathbf{\hat x}}} \right|{{{\mathbf{\hat x}}}^H}{\mathbf{\hat x}} = 1} \right\} \in {\mathbb{S}^{N - 1}}
\label{eq97}.    
\end{equation}
The Gaussian width of unit sphere is $\sqrt N $ \cite{vershynin2018high}, i.e.,
\begin{equation}
\begin{array}{l}
W\left( {\hat S} \right) = \sqrt N 
\end{array}
\label{eq98}.    
\end{equation}
Thus we can get
\begin{equation}
\begin{array}{*{20}{l}}
{E\left( {\min \left| {{\bf{\tilde G\hat x}}} \right|} \right) \ge \sqrt L  - \sqrt N }\\
{E\left( {\max \left| {{\bf{\tilde G\hat x}}} \right|} \right) \le \sqrt L  + \sqrt N }
\end{array}
\label{eq99},    
\end{equation}
Combining \eqref{eq96} and \eqref{eq99} we can get \eqref{eq33_1},
which completes the proof.

\subsection{Rate and Gradient}	
\label{appF}
The rate of each user in problem \eqref{eq76} and the corresponding gradients for ${\mathbf{v}}$ are respectively denoted as \eqref{eq75} and \eqref{eq77} at the top of this page.

\begin{figure*}[t]
\begin{equation}
\begin{array}{l}
R_k^{\left[ i \right]} = {\log _2}\left( {1 + \frac{{{{\left| {\left( {{\mathbf{A}}_{ik}^{[i]}} \right)_k^H{\mathbf{v}} + {{\left( {{\mathbf{h}}_{{B_i},k}^{[i]}} \right)}_k}} \right|}^2}p_k^{[i]}}}{{\sum\limits_{j \ne k}^K {{{\left| {\left( {{\mathbf{A}}_{ij}^{[i]}} \right)_k^H{\mathbf{v}} + {{\left( {{\mathbf{h}}_{{B_i},j}^{[i]}} \right)}_k}} \right|}^2}p_j^{[i]}}  + \sum\limits_{g \ne i}^G {\sum\limits_{j = 1}^K {{{\left| {\left( {{\mathbf{A}}_{ij}^{[g]}} \right)_k^H{\mathbf{v}} + {{\left( {{\mathbf{h}}_{{B_i},j}^{[g]}} \right)}_k}} \right|}^2}p_j^{[g]}} } {\text{ + }}{\sigma ^2}}}} \right)
\label{eq75}       
\end{array}
\end{equation}

\begin{equation}
\begin{array}{l}
\nabla W\left( {\mathbf{v}} \right){\text{  =  }}\sum\limits_{i = 1}^G {\sum\limits_{k = 1}^K {\frac{2}{{\ln 2}}\left( {\begin{array}{*{20}{l}}
  {\frac{{\sum\limits_{g = 1}^G {\sum\limits_{j = 1}^K {p_j^{[g]} \cdot {\text{ve}}{{\text{c}}_k}\left( {{\mathbf{A}}_{ij}^{[g]}} \right) \cdot \left( {\left( {{\mathbf{A}}_{ij}^{[i]}} \right)_k^H{\mathbf{v}} + {{\left( {{\mathbf{h}}_{{B_i},j}^{[i]}} \right)}_k}} \right)} } }}{{\sum\limits_{g = 1}^G {\sum\limits_{j = 1}^K {{{\left| {\left( {{\mathbf{A}}_{ij}^{[g]}} \right)_k^H{\mathbf{v}} + {{\left( {{\mathbf{h}}_{{B_i},j}^{[g]}} \right)}_k}} \right|}^2}p_j^{[g]}} } {\text{  +  }}{\sigma ^2}}}} \\ 
  { - \frac{{\sum\limits_{j \ne k}^K {p_j^{[i]} \cdot {\text{ve}}{{\text{c}}_k}\left( {{\mathbf{A}}_{ij}^{[i]}} \right) \cdot \left( {\left( {{\mathbf{A}}_{ij}^{[i]}} \right)_k^H{\mathbf{v}} + {{\left( {{\mathbf{h}}_{{B_i},j}^{[i]}} \right)}_k}} \right)}  + \sum\limits_{g \ne i}^G {\sum\limits_{j = 1}^K {p_j^{[g]} \cdot {\text{ve}}{{\text{c}}_k}\left( {{\mathbf{A}}_{ij}^{[g]}} \right) \cdot \left( {\left( {{\mathbf{A}}_{ij}^{[g]}} \right)_k^H{\mathbf{v}} + {{\left( {{\mathbf{h}}_{{B_i},j}^{[g]}} \right)}_k}} \right)} } }}{{\sum\limits_{j \ne k}^K {{{\left| {\left( {{\mathbf{A}}_{ij}^{[i]}} \right)_k^H{\mathbf{v}} + {{\left( {{\mathbf{h}}_{{B_i},j}^{[i]}} \right)}_k}} \right|}^2}p_j^{[i]}}  + \sum\limits_{g \ne i}^G {\sum\limits_{j = 1}^K {{{\left| {\left( {{\mathbf{A}}_{ij}^{[g]}} \right)_k^H{\mathbf{v}} + {{\left( {{\mathbf{h}}_{{B_i},j}^{[g]}} \right)}_k}} \right|}^2}p_j^{[g]}} } {\text{  +  }}{\sigma ^2}}}} 
\end{array}} \right)} } 
\label{eq77}   
\end{array}
\end{equation}

\end{figure*}

\bibliographystyle{IEEEtran} 
\bibliography{ref}

\end{document}